\documentclass[%
reprint,
 amsmath,amssymb,
 aps,
pra,
]{revtex4-2}
\usepackage[utf8]{inputenc}
\usepackage{amsmath, amsthm, amssymb, amsfonts}
\usepackage[makeroom]{cancel}
\usepackage{graphicx}
\usepackage{bm}
\usepackage{dsfont} 
\usepackage{mathtools}
\usepackage{booktabs}
\usepackage{thm-restate}
\usepackage[colorlinks=true, linkcolor=blue, urlcolor=blue,citecolor=blue]{hyperref}
\usepackage{orcidlink}
\usepackage{mathtools}
\usepackage{multirow}
\usepackage{array}

\DeclareMathOperator{\tr}{tr}

\newcommand{\bra}[1]{\mathinner{\langle #1|}}
\newcommand{\ket}[1]{\mathinner{|#1\rangle}}

\newcommand{\one}[0]{\mathds{1}}

\DeclareUnicodeCharacter{202C}{\^{i}}

\makeatletter
\newcommand{\vast}{\bBigg@{4}}
\newcommand{\Vast}{\bBigg@{5}}
\makeatother
\newtheorem{theorem}    {Theorem}
\newtheorem{observation}[theorem]{Observation}
\newtheorem{lemma}      [theorem]{Lemma}

\usepackage{soul,xcolor}

\begin{document}

\title{Measuring multipartite entanglement efficiently by testing symmetries}

\author{Xiaoyu Liu$^{1,2\orcidlink{0009-0009-2470-9309}}$}%
\author{Jordi Tura$^{1,2\orcidlink{0000-0002-6123-1422}}$}%
\author{Albert Rico$^{3\orcidlink{0000-0001-8211-499X}}$}%

\affiliation{$^1$$\langle aQa ^L\rangle $ Applied Quantum Algorithms, Universiteit Leiden}%
\affiliation{$^2$Instituut-Lorentz, Universiteit Leiden, P.O. Box 9506, 2300 RA Leiden, The Netherlands}%
\affiliation{$^3$GIQ - Quantum Information Group, Department of Physics, Autonomous University of Barcelona, Bellaterra 08913, Barcelona, Spain}%

\date{\today}

\begin{abstract}

Recently, a technique known as quantum symmetry test has gained increasing attention for detecting bipartite entanglement in pure quantum states. 
In this work we show that, beyond qualitative detection, a family of well-defined measures of bipartite and multipartite entanglement can be obtained with symmetry tests. We propose and benchmark several efficient methods to estimate these measures, and derive near-optimal sampling strategies for each. Despite the nonlinearity of the methods, we demonstrate that the sampling error scales no worse than $O(N_{\mathrm{tot}}^{-1/2})$ with the total number of copies $N_{\mathrm{tot}}$, which suggests experimental feasibility. 
By exploiting symmetries we compute our measures for large number of copies, and derive the asymptotic decay exponents for relevant states in many-body systems. Using these results we identify tradeoffs between estimation complexity and sensitivity of the presented entanglement measures, oriented to practical implementations.
\end{abstract}

\maketitle

\emph{Introduction}\textemdash
Entanglement is a key resource in quantum information. It enables the functioning of quantum networks~\cite{einstein1935can, munro2015quantum, wehner2018quantum}, distributed computing~\cite{buhrman2003distributed, cirac1999distributed, knorzer2025distributed}, and quantum sensing~\cite{degen2017quantum, zhang2015entanglementenhanced}.
Therefore, it is crucial to both certify and quantify the amount of entanglement present in a quantum system~\cite{elben2020crossplatform, knorzer2023crossplatform, greganti2021crossverification, vanenk2007experimental, haseler2008testing}.
Using trace polynomials, several families of bipartite~\cite{hill1997entanglement, christandl2004squashed, bennett1996concentrating, bennett1996mixedstate, vidal2002computable, vedral1997quantifying} and multipartite~\cite{coffman2000distributed, wong2001potential, beckey2021computable, liu2025generalized, meyer2002global, wei2003geometric} entanglement measures have been proposed; and several multi-copy entanglement witnesses have been derived by testing group structures~\cite{elben2020mixedstate,neven2021symmetryresolved,rico2024entanglement}. 
However, the exponential growth of the Hilbert space with the number of parties still challenges their applicability and most of these methods are case-specific.

Recently, a technique known as {\em quantum symmetry test} has been shown to efficiently detect bipartite entanglement in pure quantum states~\cite{bradshaw2023cycle, laborde2023testing, laborde2024quantum, rethinasamy2025quantum, bradshaw2023quantum}: multiple copies of a subsystem of the state in hand are symmetric under permutations, if and only if the selected system has no entanglement with the rest. This symmetry can be efficiently verified in practice with well-known methods such as the G-Bose symmetry test~\cite{barenco1997stabilization, bradshaw2023cycle, laborde2024quantum, laborde2023testing}. 
However, to the best of our knowledge, this technique is at the moment limited to qualitatively detecting whether or not a state is entangled across a selected bipartition.

In this work we show that the symmetry test leads to a family of well-defined and efficiently computable measures of pure state bipartite and multipartite entanglement, which we denote as \emph{Symmetrized Entanglement}. 
These are given by the projection of subsystem copies onto the center of the symmetric, cyclic and dihedral groups (Theorem~\ref{thrm:general-ent}). 
We also show several relevant properties of the Symmetrized Entanglement (Theorem~\ref{thrm:extreme}).
To practically estimate the Symmetrized Entanglement for these three groups, we analyze the performance of four methods: generalized SWAP test~\cite{barenco1997stabilization, buhrman2001quantum, foulds2021controlled, brun2004measuring, cotler2019quantum, huggins2021virtual, koczor2021exponential, quek2024multivariate, buhrman2001quantum, gottesman2001quantum, johri2017entanglement, ekert2002direct, yirka2021qubitefficient, subasi2019entanglement, oszmaniec2024measuring}, simultaneous moment estimation~\cite{shi2025nearoptimal}, G-Bose symmetry test~\cite{barenco1997stabilization, bradshaw2023cycle, laborde2024quantum, laborde2023testing}, and cyclic permutation test~\cite{buhrman2024permutation, kada2008efficiency, beckey2021computable,liu2025generalized} (Fig.~\ref{fig:circuits}). 
We provide near-optimal sample allocation strategies with respect to the number of state copies (Table~\ref{tab:alloc}),
and we show that all four approaches for estimating Symmetrized Entanglement yield sampling error scaling no worse than $O(N_{\mathrm{tot}}^{-1/2})$ (Fig.~\ref{fig:error} and Observation~\ref{obs:scaling}).
Finally, we compute our measures for large number of copies and identify distinct exponential decay rates of the symmetry test values for GHZ- and Dicke-like families of states (Eqs.~\eqref{eq:LimGHZmain},~\eqref{eq:LimWmain} and Fig.~\ref{fig:theta-state}).
This leads to different tradeoffs between the estimation complexity and sensitivity of our entanglement measures associated with different permutation groups.

\emph{Symmetrized Entanglement}\textemdash
Consider an $n$-partite quantum state $\ket{\psi}\in(\mathbb{C}^d)^{\otimes n}$ and a subset $S$ of $|S| < n$ parties with complementary $S^c$, so that $|S| + |S^c| = n$. Denote the symmetric group of permutations of $k$ elements as $\mathcal{S}_k$. Each permutation $\pi\in \mathcal{S}_k$ acts on $k$ copies of the subsystem $S$ via $\pi_S\ket{v_1}\otimes\cdots\otimes\ket{v_k}=\ket{v_{\pi^{-1}(1)}}\otimes\cdots\otimes\ket{v_{\pi^{-1}(k)}}\in(\mathbb{C}^{d^{|S|}})^{\otimes k}$. Given a permutation subgroup $\mathcal{G}_k\subseteq \mathcal{S}_k$, we consider the subspace of  $(\mathbb{C}^{d^{|S|}})^{\otimes k}$ (namely $k$ copies of $S$) that is symmetric under the action of $\mathcal{G}_k$. The projector onto this invariant subspace is given by $P_{k}^S(\mathcal{G}) = |\mathcal{G}_k|^{-1}\sum_{\pi\in\mathcal{G}_k}\pi_S$, where $|\mathcal{G}_k|$ is the number of elements in $\mathcal{G}_k$. 

In recent works, the following projection has been recently shown to be effective at detecting entanglement~\cite{bradshaw2023cycle, laborde2024quantum, laborde2023testing}:
\begin{equation}\label{eq:general-C}
    C_k^S(\ket{\psi},\mathcal{G}) := \bra{\psi}^{\otimes k} \left( P_{k}^S(\mathcal{G}) \otimes \mathbb{I}_k^{S^c} \right) \ket{\psi}^{\otimes k},
\end{equation}
where $\mathbb{I}_k^{S^c}$ is the identity acting on $(\mathbb{C}^{d^{|S_c|}})^{\otimes k}$ (namely $k$ copies of the subsystem $S^c$). 
The subscript $k$ on $\mathcal{G}$ is omitted in Eq.~\eqref{eq:general-C} for simplicity as $k$ is already specified, and we apply this convention throughout.
This projection onto the $\mathcal{G}_k$-symmetric subspaces, also named {\em acceptance probability}, has recently gained attention as a separability test: it takes unit value if and only if $\ket{\psi}$ is separable across $S|S^c$~\cite{bradshaw2023cycle,laborde2023testing,bradshaw2023quantum}. The projection in Eq.~\eqref{eq:general-C} has been considered for exemplary groups~\cite{bradshaw2023cycle,laborde2023testing}. Here we show that for the symmetric, cyclic and dihedral groups, it provides a family of pure state entanglement measures:
\begin{theorem} (Symmetrized Entanglement)
Let $\mathcal{S}_k$, $\mathcal{C}_k$ and $\mathcal{D}_k$ denote the symmetric, cyclic and dihedral $k$-element permutation groups, respectively. If $\mathcal{G}_k$ is one of these groups, then the following statements hold:\\
    1. The quantity
    \begin{equation}\label{eq:BIPmeasure}
        \mathcal{E}_k^S(\ket{\psi}, \mathcal{G}) := 1 - C_k^S(\ket{\psi},\mathcal{G})
    \end{equation}
    is a valid measure of pure state bipartite entanglement between subsystems $S$ and $S^c$.\\
    2. Averaging over all bipartitions with $|S| = s$ local parties yields a valid multipartite entanglement measure,
    \begin{equation}\label{eq:AVGmeasure}
        \mathcal{E}_k^s(\ket{\psi},\mathcal{G}) := 1-C_k^s(\ket{\psi},\mathcal{G}),
    \end{equation}
    where we denote $C_k^s(\ket{\psi},\mathcal{G})=\binom{n}{s}^{-1}\hspace{-2pt}\sum_{|S| = s}\hspace{-2pt} C_k^S(\ket{\psi},\mathcal{G})$.\\
    3. Maximizing $C_k^S(\ket{\psi},\mathcal{G})$ over all partitions $S$ vs $S^c$ yields a genuinely multipartite entanglement measure,
    \begin{equation}\label{eq:GMEmeasure}
        \mathcal{E}_{\mathrm{GME},k}(\ket{\psi},\mathcal{G}) := 1 - \max_{S : \, |S| + |S^c| = n} C_k^S(\ket{\psi},\mathcal{G}).
    \end{equation}
\label{thrm:general-ent}
\end{theorem}
We denote this family of entanglement measures as \emph{Symmetrized Entanglement}. Here Eq.~\eqref{eq:AVGmeasure} measures the average entanglement of $\ket{\psi}$ in Eq.~\eqref{eq:BIPmeasure},
and Eq.~\eqref{eq:GMEmeasure} measures genuinely multipartite entanglement, in the sense that it vanishes unless $\ket{\psi}$ is entangled across all bipartitions. 
In different contexts, similar approaches measuring multipartite and genuine entanglement have been proposed in~\cite{scott2004multipartite} and~\cite{ma2011measure}.
The proof of Theorem~\ref{thrm:general-ent} is given in Supplemental Material~\ref{sm:monotone}.

\emph{Efficient computation}\textemdash
Obtaining Eq.~\eqref{eq:general-C} is in general challenging, especially for the full symmetric group $\mathcal{G}_k=\mathcal{S}_k$ where $k!$ permutations are involved. 
Therefore, current effort is being devoted to finding efficient ways to compute the Symmetrized Entanglement measures. 
For that we use the cycle notation $\pi=(\alpha)(\beta)\cdots(\gamma)\in\mathcal{G}_k$, where $\pi$ has $m_l$ cycles of length $l$ and thus $\sum_{l}lm_l=k$. This defines the conjugacy class of $\pi$, given by a partition $\vec{\mathbf{k}}:=[1^{m_1}2^{m_2}\cdots t^{m_t}]\vdash k$. 
Using that $\tr(\rho_S^{\otimes l}(1 \ 2 \ \cdots \ l)_S)=\tr(\rho_S^l)$, one can express Eq.~\eqref{eq:general-C} as~\cite{bradshaw2023quantum,bradshaw2023cycle,laborde2023testing}:
\begin{equation}
    C_k^S(\ket{\psi},\mathcal{G})
    = \frac{1}{|\mathcal{G}_k|}\sum_{\vec{\mathbf{k}} \vdash k} N_{\mathcal{G}}(\vec{\mathbf{k}}) \prod_{l}\left[\tr(\rho_S^l)\right]^{m_l},
\label{eq:general-partition}
\end{equation}
where $N_{\mathcal{G}}(\vec{\mathbf{k}})$ is the number of permutations of cycle type $\vec{\mathbf{k}}$ in group $\mathcal{G}_k$.
For $\mathcal{G}_k=\mathcal{S}_k$ 
the number of terms to be evaluated, i.e., the number of different partitions of $k$, scales as $p(k)\sim\Theta( \exp( \pi \sqrt{2k/3}  ))$~\cite{hardy1918asymptotic,andrews1998theory},
which significantly reduces the brute force number of computations $k!$.
Further simplification can also be achieved using generating functions, leading to a recurrence relation for $C_k^S(\ket{\psi},\mathcal{S})$ in $k$.
The details of these reductions, together with the simplified expressions for groups $\mathcal{C}$ and $\mathcal{D}$~\cite{bradshaw2023cycle,laborde2023testing,bradshaw2023quantum}, are listed in Supplemental Material~\ref{sm:exemplary-cases}.

For our purposes we also express Eq.~\eqref{eq:general-C} in terms of the eigenvalues $\lambda_i$ of $\rho_S$:
\begin{equation}
C_k^S(\ket{\psi}, \mathcal{G}) = \sum_{
g_1+\cdots+g_r=k
} a_{\mathcal{G}_k}(g_1,\cdots,g_r) \lambda_1^{g_1} \cdots \lambda_r^{g_r},
\label{eq:main-text-spectra}
\end{equation}
where the summation is done over nonnegative integers $g_i$ summing to $k$ and $r$ is the rank of $\rho_S$. 
The coefficients $a_{\mathcal{G}_k}(g_1,...,g_r)$ depend on the group $\mathcal{G}_k$ in consideration and their exact form is given in Supplemental Material~\ref{sm:part-moments-spectra}. 
Particularly, for the symmetric group $\mathcal{S}_k$, we obtain $a_{\mathcal{S}_k}(g_1,...,g_r)=1$. 
This computation is most efficient for states with low-rank reductions, such as matrix product states describing short-range many-body interactions~\cite{zhang2015entanglementenhanced}. 
It also allows us to prove Theorems~\ref{thrm:general-ent} and~\ref{thrm:extreme}, as detailed in Supplemental Material~\ref{sm:monotone} and~\ref{sm:thrm6}. 

\begin{figure}
    \centering
    \includegraphics[width=0.85\linewidth]{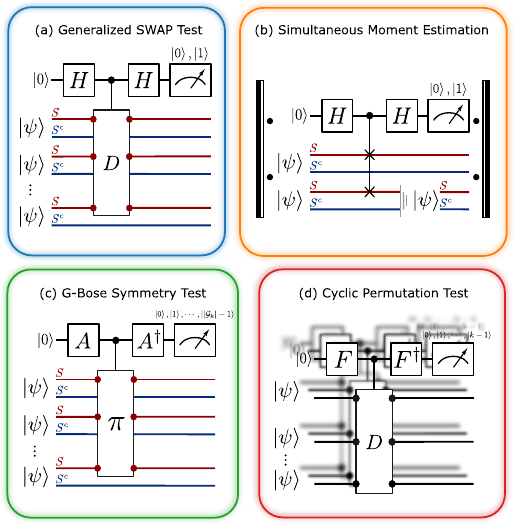}
    \caption{Circuit diagrams of (a) generalized SWAP test, (b) simultaneous moment estimation, (c) G-Bose symmetry test and (d) cyclic permutation test. Here $A$ is any gate mapping $\ket{0}$ to a coherent equal superposition, $F$ is the qudit Fourier transform, $D$ is a full-cycle permutation and $\pi$ are the permutations in $\mathcal{G}_k$.}
    \label{fig:circuits}
\end{figure}

\emph{Properties of Symmetrized Entanglement}\textemdash
The symmetries of the entanglement measures introduced in this work enable us to characterize their behavior concerning extremal cases and asymptotic limits, as follows.
\begin{theorem}(Properties) Let $k$ be the number of copies used, $S$ a certain subsystem, and $\mathcal{G}_k$ a $k$-element permutation subgroup. Then the following holds:\\
1. The Symmetrized Entanglement $\mathcal{E}_k^S(\ket{\psi}, \mathcal{G})$ reaches its maximum value when $\rho_S$ is maximally mixed, $\rho_S=\one/d^{|S|}$.\\ 
2. If $\mathcal{G}_k$ is either $\mathcal{S}_k$, $\mathcal{C}_k$ or $\mathcal{D}_k$, then
\begin{equation}
    \lim_{k\rightarrow\infty} \mathcal{E}_k^S(\ket{\psi}, \mathcal{G})=\delta_{\text{PR}},
\end{equation}
where $\delta_{\text{PR}}$ is $0$ if $\ket{\psi}=\ket{\phi}_S\otimes\ket{\varphi}_{S^c}$ and $1$ otherwise.\\ 
3. The following chain of inequalities holds for any fixed $k$, $S$ and $\ket{\psi}$:
\begin{equation}
1-\tr(\rho_S^k)\geqslant\mathcal{E}_k^S(\ket{\psi}\hspace{-3pt}, \hspace{-1pt}\mathcal{S})\geqslant\mathcal{E}_k^S(\ket{\psi}\hspace{-3pt},\hspace{-1pt} \mathcal{D})\geqslant\mathcal{E}_k^S(\ket{\psi}\hspace{-3pt},\hspace{-1pt} \mathcal{C}).
\end{equation}
\label{thrm:extreme}
\end{theorem}
Theorem~\ref{thrm:extreme}.1 implies that both the average multipartite and genuinely multipartite families of Symmetrized Entanglement in Eqs.~\eqref{eq:AVGmeasure} and~\eqref{eq:GMEmeasure} are maximal when all subsystems of size $s$ are maximally entangled to the rest, in which case $\ket{\psi}$ is called $s$-uniform~\cite{goyeneche2015absolutely}. Yet, this maximum value is not always saturated since the existence of $s$-uniform states is not guaranteed for all system sizes~\cite{scott2004multipartite, helwig2013absolutely, huber2017absolutely}.
We also list the exact upper bounds for the groups $\mathcal{S}_k$, $\mathcal{C}_k$ and $\mathcal{D}_k$ in Supplemental Material~\ref{sm:thrm6}. 
Theorem~\ref{thrm:extreme}.2 implies that $\lim_{k\rightarrow\infty} \mathcal{E}_k^s(\ket{\psi}, \mathcal{G})=1$ if and only if $\ket{\psi}$ is not fully separable, thus serving as a tunable test for multipartite entanglement through $k$.
Theorem~\ref{thrm:extreme}.3 analytically proves a property that was observed numerically in~\cite{bradshaw2023cycle}.
The first inequality of the chain implies that our Symmetrized Entanglement can be upper bounded by $1-\tr(\rho_S^k)$, which is an entanglement monotone known as the $q$-concurrence~\cite{yang2021parametrized}. 
The detailed proof of Theorem~\ref{thrm:extreme} can be found in Supplemental Material~\ref{sm:thrm6}.

\begin{figure}
    \centering
    \includegraphics[width=0.95\linewidth]{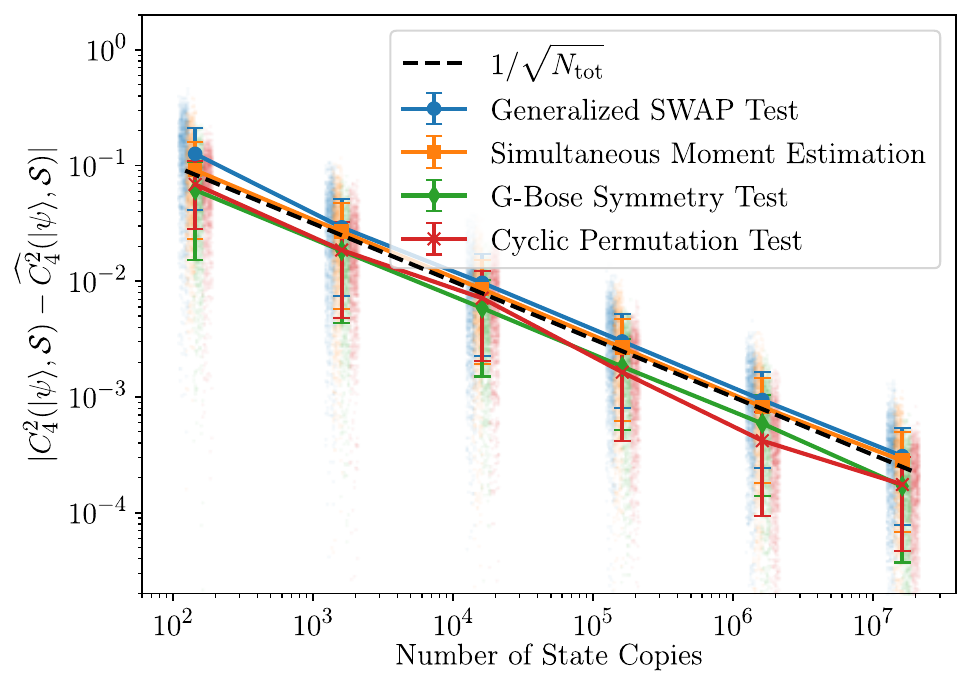}
    \caption{Absolute sampling error in estimating multipartite $C_4^{2}(\ket{\psi}, \mathcal{S})$. 
    For each circuit and each total copy budget $N_{\mathrm{tot}}$, resulting errors are averaged over $1000$ 4-qubit Haar-random pure states (the individual errors are shown as scatter points; within each cluster the points share the same $N_{\mathrm{tot}}$ and are slightly offset horizontally for visual clarity). 
    The empirical errors exhibit the scaling $\varepsilon \sim N_{\mathrm{tot}}^{-1/2}$. }
    \label{fig:error}
\end{figure}

\emph{Estimation in quantum circuits}\textemdash
We will now assess how the entanglement measures in Theorem~\ref{thrm:general-ent} and their properties in Theorem~\ref{thrm:extreme} can be obtained and tested in practical use. 
For that we will consider four different circuit diagrams to estimate the Symmetrized Entanglement in practice, and provide a detailed analysis for each. The most straightforward approach is G-Bose symmetry test (Fig.~\ref{fig:circuits}(c)), since it directly computes Eq.~\eqref{eq:general-C}.
The auxiliary qudit of dimension $|\mathcal{G}_k|$ is initialized in an equal superposition of all its levels, $A\ket{0}=|\mathcal{G}_k|^{-1}\sum_{j=0}^{|\mathcal{G}_k|-1}\ket{j}$, where $A$ is a unitary gate in dimension $|\mathcal{G}_k|$ creating a coherent superposition. One canonical example of $A$ is the Fourier transform $F=\sum_{p,q=0}^{|\mathcal{G}_k|-1}\omega^{pq}\ket{p}\bra{q}$ with $\omega=e^{2\pi i/|\mathcal{G}_k|}$, but efficiency can be gained with other choices. 
The controlled-$\pi$ gate $\sum_{\pi\in\mathcal{G}_k}\ket{j(\pi_S)}\bra{j(\pi_S)}\otimes \pi_S$ is then applied between the auxiliary qudit and the copies of the system $S$, where $j(\pi_S)\in\{0,1,\cdots,|\mathcal{G}_k|-1\}$ labels distinct coherent controls corresponding to different permutations $\pi_S$.
Afterwards, $A^{\dagger}$ is applied to the auxiliary qudit, which is then measured.
The probability of obtaining the outcome $\ket{0}$ is precisely $C_k^S(\ket{\psi}, \mathcal{G})$.
Although the coherent control is challenging for high dimensions $|\mathcal{G}_k|$, it has been shown that for the groups $\mathcal{S}_k$ and $\mathcal{C}_k$, the operation can be implemented using qubit circuits with $O(k^2)$ and $O(k\log k)$ controlled-SWAP gates respectively~\cite{barenco1997stabilization, bradshaw2023cycle}, and even realized through single-qubit measurements~\cite{laborde2024quantum}.
For the cyclic group $\mathcal{C}_k$, one can also use the parallelized cyclic permutation test shown in Fig.~\ref{fig:circuits}(d) to estimate $C_k^S(\ket{\psi},\mathcal{C})$ for arbitrary subsystem $S$ simultaneously, with only different classical postprocessing of the measurement outcomes~\cite{liu2025generalized}. 
For that one uses the $k$-dimensional Fourier transform $F$ and the controlled-$D$ gate is $\sum_{j=0}^{k-1}\ket{j}\bra{j}\otimes D^j$ with the full-cycle permutation $D = (1\ 2\ \cdots\ k)$.

Alternatively, from Eq.~\eqref{eq:general-partition}, $C_k^S(\ket{\psi}, \mathcal{G})$ can be also computed by applying multiple generalized SWAP tests in Fig.~\ref{fig:circuits}(a)~\cite{barenco1997stabilization, buhrman2001quantum, foulds2021controlled, brun2004measuring, cotler2019quantum, huggins2021virtual, koczor2021exponential, quek2024multivariate, buhrman2001quantum, gottesman2001quantum, johri2017entanglement, ekert2002direct, yirka2021qubitefficient, subasi2019entanglement, oszmaniec2024measuring} to acquire each required moment $\tr(\rho_S^j)$.
Moreover, rather than estimating $\tr(\rho_S^j)$ for each pair of $S$ and $j$ separately, the parallelized cyclic permutation test in Fig.~\ref{fig:circuits}(d) enables simultaneous estimation of state moments for a fixed $j$ over arbitrary subsystems $S$~\cite{liu2025generalized, beckey2021computable}.
Conversely, for a fixed subsystem $S$, one can also simultaneously estimate all moments $\tr(\rho_S^j)$ for $2\leqslant j\leqslant k$ by resetting and regenerating $\ket{\psi}$ on the same registers, i.e., by repeating the circuit in Fig.~\ref{fig:circuits}(b) $k-1$ times~\cite{shi2025nearoptimal}.

Note that for the diagrams in Fig.~\ref{fig:circuits}(a) and (d), estimating the components of $C_k^S(\ket{\psi}, \mathcal{G})$ may require several distinct subcircuits.
We derive near-optimal strategies for allocating the total number of executions across these subcircuits by combining Hoeffding’s inequality with a Lagrange-multiplier optimization, summarized in Table~\ref{tab:alloc}.
We also report numerical results for the absolute estimation error of $C_4^2(\ket{\psi}, \mathcal{S})$ for 1000 Haar random pure states as a function of $N_{\mathrm{tot}}$ in Fig.~\ref{fig:error}. 
Based on the numerical results, we observe the following:
\begin{observation}
    The sampling errors $\varepsilon$ in estimating both $C_k^S(\ket{\psi}, \mathcal{G})$ and $C_k^s(\ket{\psi}, \mathcal{G})$ for the groups $\mathcal{S}_k,\mathcal{C}_k$ and $\mathcal{D}_k$ scale as $\varepsilon\sim N_{\mathrm{tot}}^{-1/2}$ in all sample instances.
\label{obs:scaling}
\end{observation}

This scaling is also consistent with our analytical results, although the derived bound includes a constant factor greater than one.
Note that this scaling is natural for G-Bose symmetry test (Fig.~\ref{fig:circuits}(c)), which directly estimates Eq.~\eqref{eq:general-C} via binary outcomes (measurement outcomes $\ket{0}$ vs. non-$\ket{0}$),
and is already commonly used in estimating Eq.~\eqref{eq:general-C}~\cite{bradshaw2023cycle,laborde2024quantum,laborde2023testing}.
In contrast, for the methods in Fig.~\ref{fig:circuits} (a,b,d), their sampling error scalings in estimating Eq.~\eqref{eq:general-C} have not been systematically characterized to the best of our knowledge. 
Therefore, in principle one may expect the nonlinearity of Eq.~\eqref{eq:general-C} to be detrimental for the shot-noise error propagation. 
Surprisingly, our analysis suggests that these methods also exhibit the statistical $O( N_{\mathrm{tot}}^{-1/2})$ scaling. 
Namely that, combining the estimated sub-components of Eq.~\eqref{eq:general-C} does not induce substantial propagated error in the resulting estimated quantity. 
The detailed mechanisms of these circuits, derivations of the allocation rules, sampling error analysis, and additional numerical results for all three groups, as well as for the logarithmic (relative) sampling error (which is also $\sim O(N^{-1/2}_{\mathrm{tot}})$) are provided in Supplemental Material~\ref{sm:noise}.
There we also investigate how the error depends on the order $k$ using extrapolations of higher state moments via the Newton–Girard method~\cite{shin2025resourceefficient}, which reconstructs the spectra of $\rho_S$ with rank $r$ from its first $r$ state moments.

\begin{table}[tbp]
\centering
\resizebox{\columnwidth}{!}{%
\begin{tabular}{|c|l|c|}
\hline
$C_k^S(\ket{\psi}, \mathcal{G})$ & \multicolumn{1}{c|}{\textbf{Near-Optimal Allocation}} & \textbf{Circuit} \\ \hline
$C_k^S(\ket{\psi}, \mathcal{S})$  & \begin{tabular}[c]{@{}l@{}}$\{N_j\}_{j=2}^{k}$ for $\tr(\rho_S^j)$.\\ $N_j\propto j^{-4/3}$.\end{tabular}                                                                                                     & Fig.~\ref{fig:circuits}(a,d) \\ \hline
\multirow{3}{*}{$C_k^S(\ket{\psi}, \mathcal{C})$} & \begin{tabular}[c]{@{}l@{}}$\{N_j\}_{j=2}^{k}$ for $\tr(\rho_S^j)$.\\ $N_j\propto \left(\varphi(j)\right)^{2/3} j^{-4/3}\boldsymbol{\delta}(j | k)$.\end{tabular}                                                           & Fig.~\ref{fig:circuits}(a)       \\ \cline{2-3} 
                                                  & \begin{tabular}[c]{@{}l@{}}$N_k$ for $C_k^S(\ket{\psi}, \mathcal{C})$.\\ $N_k=N_{\mathrm{tot}}/k$.\end{tabular}                                                                               & Fig.~\ref{fig:circuits}(d)       \\ \hline
\multirow{4}{*}{$C_k^S(\ket{\psi}, \mathcal{D})$} & \begin{tabular}[c]{@{}l@{}}$\{N_j\}_{j=2}^{k}$ for $\tr(\rho_S^j)$.\\ $N_j\propto \left( \frac{\varphi(j)}{2j^2}\boldsymbol{\delta}(j|k) + \frac{k-1}{8}\boldsymbol{\delta}(j=2)\right)^{2/3}$.\end{tabular} & Fig.~\ref{fig:circuits}(a)       \\ \cline{2-3} 
                                                  & \begin{tabular}[c]{@{}l@{}}$N_2$ for $\tr(\rho_S^2)$; $N_k$ for $C_k^S(\ket{\psi}, \mathcal{C})$. \\  $N_2/N_k=\left( \frac{k(k-1)}{2} \right)^{2/3}$.\end{tabular}                                                  & Fig.~\ref{fig:circuits}(d)       \\ \hline
\end{tabular}}
\caption{Allocation of circuit executions across subcircuits for the generalized SWAP test (Fig.~\ref{fig:circuits}(a)) and the cyclic-permutation test (Fig.~\ref{fig:circuits}(d)) used to estimate sub-components of $C_k^S(\ket{\psi},\mathcal{G})$. 
These sub-components follow from the simplified expressions of $C_k^S(\ket{\psi},\mathcal{G})$ for each exemplary group $\mathcal{G}_k$, as shown in Supplemental Material~\ref{sm:exemplary-cases}.
For a subcircuit that consumes $j$ copies, $N_j$ denotes its number of executions. 
Thus, the total copy budget is $N_{\mathrm{tot}}=\sum_{j=2}^{k} j\,N_j$. 
We use $\boldsymbol{\delta}(\cdot)$ for the Kronecker delta and $j|k$ denotes that $j$ divides $k$. 
For the approaches in Fig.~\ref{fig:circuits}(b,c), the allocation always reduces to $N_k=N_{\mathrm{tot}}/k$ as no multiple subcircuits are required for a given $S$.}
\label{tab:alloc}
\end{table}

\begin{figure}
    \centering
    \includegraphics[width=0.95\linewidth]{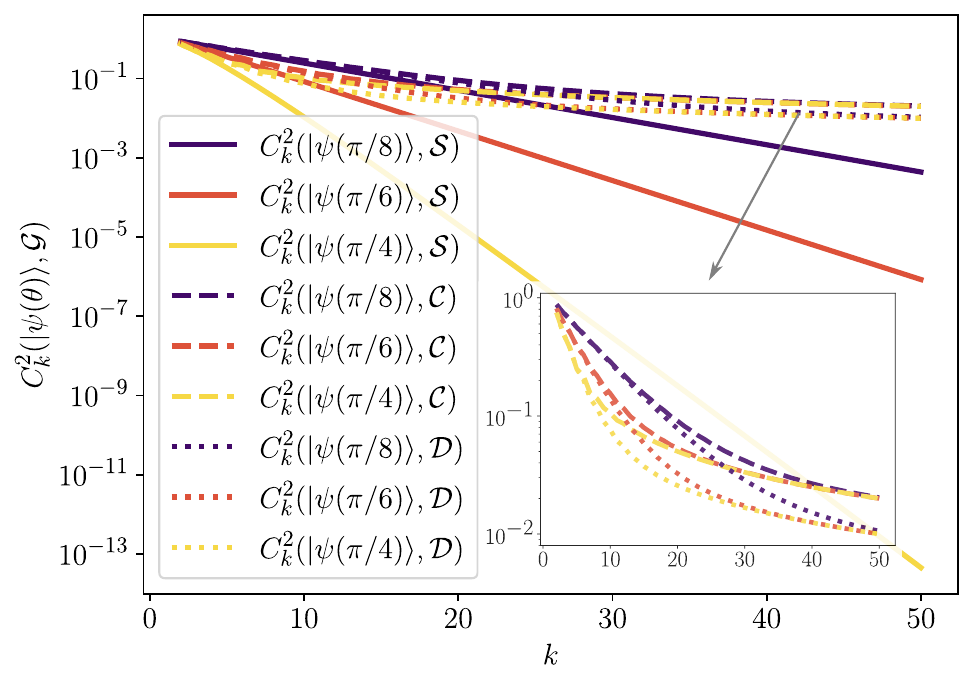}
    \caption{Values of $C_k^{2}(\ket{\psi(\theta)},\mathcal{G})$ for $\theta\in\{\pi/8,\pi/6,\pi/4\}$, $\mathcal{G}_k=\mathcal{S}_k,\mathcal{C}_k,\mathcal{D}_k$ and $k=2,\cdots,50$. 
    For symmetric projection onto $\mathcal{S}_k$, one observes a clear $\theta$-dependent exponential decay with $k$. 
    In contrast, the differences for the cyclic and dihedral projections onto $\mathcal{C}_k$ and $\mathcal{D}_k$ are much more subtle (see zoom).}
    \label{fig:theta-state}
\end{figure}

\emph{Examples}\textemdash
Here we will evaluate analytically the Symmetrized Entanglement for relevant families of quantum states in many-body systems. 
We first consider the $GHZ_{\theta}$ state, $\ket{\psi(\theta)}=\sin\theta \ket{0}^{\otimes n}+ \cos\theta \ket{1}^{\otimes n}$~\cite{walck2008only}.
The comparisons of $C_k^s(\ket{\psi(\theta)}, \mathcal{G})$ for different angles $\theta$ are shown in Fig.~\ref{fig:theta-state}, where one observes a clear $\theta$-dependent exponential decay of $C_k^s(\ket{\psi(\theta)},\mathcal{S})$ with $k$. 
Indeed, in the limit of large $k$ we analytically show the following asymptotic behavior: 
\begin{equation}\label{eq:LimGHZmain} 
\lim_{k\rightarrow\infty}\frac{C_{k+1}^s(\ket{\psi(\theta)},\mathcal{S})}{C_{k}^s(\ket{\psi(\theta)},\mathcal{S})}=\max(\sin^2\theta,\cos^2\theta).
\end{equation}
Notice from Eq.~\eqref{eq:LimGHZmain} that, within the family of $GHZ_{\theta}$ states, the Symmetrized Entanglement is maximized for $\ket{\psi(\pi/4)}=\ket{GHZ}$.
Differences in the decay rates across the cyclic $\mathcal{C}_k$ and dihedral $\mathcal{D}_k$ groups are also present, although they are more subtle since the decay is comparatively slower than for the symmetric group $\mathcal{S}_k$.  
This implies that estimating $C_k^s(\ket{\psi},\mathcal{S})$ is more sensitive than $C_k^s(\ket{\psi},\mathcal{D})$ or $C_k^s(\ket{\psi},\mathcal{C})$ at distinguishing states with similar amounts of entanglement.
However, its practical estimation is more experimentally demanding:
it requires higher-dimensional coherent controls (Fig.~\ref{fig:circuits}(c)) or additional state moment estimates (Fig.~\ref{fig:circuits}(a,b,d)),
and accurately acquiring such low $C_k^S(\ket{\psi}, \mathcal{S})$ for large $k$ also typically demands more state copies even with Newton-Girard extrapolations~\cite{shin2025resourceefficient}.
This provides a tradeoff between sensitivity and estimation complexity of the different entanglement measures introduced in this work.

As a second example, we consider the Dicke state 
$\ket{D(n,e)}=\binom{n}{e}^{-1}\sum_{ x\in\{0,1\}^{n}, \ w(x)=e } \ket{x}$ where $w(x)$ is the Hamming weight of the excitation string $x$~\cite{dicke1954coherence,marconi2025symmetric}.
A notable example is the W-state $\ket{W}=\ket{D(n,1)}$~\cite{dur2000three}. 
For this case we analytically show that the asymptotic behavior is given as follows,
\begin{equation}\label{eq:LimWmain}
\lim_{k\rightarrow \infty}\frac{C_{k+1}^s(\ket{W}, \mathcal{S})}{C_k^s(\ket{W}, \mathcal{S})} = \max\left(\frac{n-s}{n}, \frac{s}{n} \right),
\end{equation}
which depends only on the number of qubits $n$ and the size of the subsystem $s=|S|$. 
Moreover, for the groups $\mathcal{S}_k$, $\mathcal{D}_k$ and $\mathcal{C}_k$, we show that $\mathcal{E}_k^s(\ket{W},\mathcal{G}) \leqslant \mathcal{E}_k^s(\ket{GHZ},\mathcal{G})$ holds for arbitrary $1\leqslant s \leqslant n-1$, even though the $\ket{GHZ}$ is only 1-uniform and thus does not maximize the Symmetrized Entanglement for $s\neq 1$ according to Theorem~\ref{thrm:extreme}. 

For more general Dicke states, and for $1 \leqslant e<e+1\leqslant \lfloor n/2\rfloor$, we show that:
\begin{equation}
\mathcal{E}_k^1(\ket{D(n,e)},\mathcal{G})\leqslant\mathcal{E}_k^1(\ket{D(n,e+1)},\mathcal{G}).
\end{equation}
Note that, due to the permutation symmetry of the $GHZ_{\theta}$ and Dicke states, all of the above statements also apply to the bipartite case with subsystem cardinality $|S|=s$. 
Moreover, the monotonic non-increase of Eq.~\eqref{eq:general-C} with respect to $k$ has been proven for group $\mathcal{S}$ in arbitrary cases, and has also been observed for groups $\mathcal{C}$ and $\mathcal{D}$, although this remains unproven in the latter cases~\cite{bradshaw2023cycle}.
The detailed analysis of these examples, along with numerical simulations using the estimation methods in Fig.~\ref{fig:circuits} and the proof of the $k$-monotone non-increasing behavior of $C_k^S(\ket{\psi},\mathcal{S})$, is provided in Supplemental Material~\ref{sm:example}.

\emph{Conclusions and Outlook}\textemdash 
In this work we have shown that symmetry tests enable to quantitatively measure the amount of entanglement present in a quantum system. We have introduced the Symmetrized Entanglement, a family of well-defined entanglement measures for pure states, defined by testing symmetry under the symmetric, cyclic, and dihedral groups. 
Conceptually, this result provides the symmetry tests with a clear operational interpretation. For practical purposes, it reinforces and widely extends the applicability of symmetry tests beyond qualitative bipartite entanglement detection. 
The Symmetrized Entanglement also applies to measuring (genuinely) multipartite entanglement, which is a current important challenge.
Moreover, we have established useful properties of the Symmetrized Entanglement under these three group projections, which characterize their behavior especially in the limit of large systems and copies.

To facilitate an experimental implementation of the entanglement measures derived in this work, 
we have analyzed practical measurement schemes on quantum hardware using four circuit families: generalized SWAP test, simultaneous moment estimation, G-Bose symmetry test, and cyclic permutation test. 
Among these, the cyclic permutation test enables simultaneous estimation of both bipartite and multipartite Symmetrized Entanglement across arbitrary subsystems. 
We have analyzed the sampling noise for all four approaches, both numerically and analytically. 
We have provided near-optimal allocations of circuit executions, and demonstrated that the estimation error scales as $\varepsilon\sim N^{-1/2}_{\mathrm{tot}}$ despite the highly nonlinear nature in some of the estimation approaches. 
Finally, we have tested our measures on GHZ- and Dicke- state families, where we have found a different exponential decay of the acceptance probability for differently entangled states with the order $k$. 
These explicit results provide tradeoffs between the complexity and sensitivity of practically estimating the entanglement associated with different permutation groups.

Based on our theoretical results and experimental simulation, we consider that the most important next step is to estimate the entanglement introduced here through the proposed circuits in realistic quantum hardware. 
Nevertheless, several important questions still remain at the theoretical level. 
In particular, we conjecture that the Symmetrized Entanglement is a well-defined entanglement measure for arbitrary subgroups of the symmetric group. 
For a general permutation group, the circuit constructions in Fig.~\ref{fig:circuits} apply similarly, and analogous sample allocation strategies can also be derived. 
This conjecture is motivated for practical purposes: it would be ideal to identify permutation subgroups whose asymptotic behavior of $C_k^S(\ket{\psi},\mathcal{G})$ (or $C_k^s(\ket{\psi},\mathcal{G})$) lies between that of $\mathcal{S}_k$ and $\mathcal{C}_k,\mathcal{D}_k$, to obtain more refined tradeoffs between estimation complexity and sensitivity.
Concerning different groups and states, it would be interesting to understand the associated exponential decays in the context of many-body dynamical systems.

\emph{Acknowledgments}\textemdash 
We thank Jinfu Chen, Mengyao Hu, Owidiusz Makuta, Anastasiia Skurativska and Zherui Jerry Wang for insightful discussions.
The numerical experiments of this work were performed using the compute resources from the Academic Leiden Interdisciplinary Cluster Environment (ALICE) provided by Leiden University.
J.T. acknowledges the support received from the European Union’s Horizon Europe research and innovation programme through the ERC StG FINE-TEA-SQUAD (Grant No. 101040729).
This work is supported by the Dutch National Growth Fund (NGF), as part of the Quantum Delta NL programme.
This work is part of the ‘Quantum Inspire – the Dutch Quantum Computer in the Cloud’ project (with project number [NWA.1292.19.194]) of the NWA research program ‘Research on Routes by Consortia (ORC)’, which is funded by the Netherlands Organization for Scientific Research (NWO). 
A. R. acknowledges financial support from Spanish MICIN (projects: PID2022:141283NBI00;139099NBI00) with the support of FEDER funds, the Spanish Government with funding from European Union NextGenerationEU (PRTR-C17.I1), the Generalitat de Catalunya, the Ministry for Digital Transformation and of Civil Service of the Spanish Government through the QUANTUM ENIA project -Quantum Spain Project- through the Recovery, Transformation and Resilience Plan NextGeneration EU within the framework of the Digital Spain 2026 Agenda.
The views and opinions expressed here are solely those of the authors and do not necessarily reflect those of the funding institutions. Neither of the funding institutions can be held responsible for them.

\bibliography{Symm_E}

\onecolumngrid
\newpage

\begingroup
\centering
{\fontsize{13pt}{10pt}\selectfont\bfseries Supplemental Material\par}
\vspace{1.5em}
\endgroup

In this Supplemental Material, we present the technical details referenced in the main text and provide additional numerical results to support our work.
For simplicity, in this Supplemental Material we will sometimes omit $\mathcal{G}$ in $\mathcal{E}_k^S(\ket{\psi},\mathcal{G})$, $C_k^S(\ket{\psi},\mathcal{G})$ or $P_k^S(\mathcal{G})$ when $\mathcal{G}$ is already clearly specified or denoted for arbitrary permutation group.

\section{Proof of Theorem~\ref{thrm:general-ent}~\label{sm:monotone}}

To prove Theorem~\ref{thrm:general-ent}, we first establish that $\mathcal{E}_k^S(\ket{\psi})$ for $\mathcal{G}=\mathcal{S},\mathcal{C},\mathcal{D}$ is a well-defined bipartite pure-state entanglement measure.
This requires verifying the following key properties shown in Lemma~\ref{lemma:sep},~\ref{lemma:locc} and~\ref{lemma:continuous}:
\begin{lemma}
    $\mathcal{E}_k^S(\ket{\psi}) = 0$ if and only if $\ket{\psi}$ is separable across the bipartition $S|S^c$.
    \label{lemma:sep}
\end{lemma}
\begin{proof}
From Eq.~\eqref{eq:general-partition}, we have:
\begin{equation}
    \mathcal{E}_k^S(\ket{\psi}) = 1 - C_k^S(\ket{\psi}) = 1 - \frac{1}{|\mathcal{G}_k|}\sum_{\vec{\mathbf{k}} \vdash k} N_{\mathcal{G}}(\vec{\mathbf{k}}) \prod_{l}\left[\tr(\rho_S^l)\right]^{m_l}.
\end{equation}
From this expression we see that $\mathcal{E}_k^S(\ket{\psi}) = 0$ if and only if $\ket{\psi}$ is separable across the bipartition $S|S^c$.  
In the \emph{if} case, the projector $P_k^S$ acts trivially on the $k$ copies of $\ket{\psi}$, yielding $C_k^S(\ket{\psi}) = 1$.
In the \emph{only if} case, assume $\tr(\rho_S^l)=1$ for all $l\geqslant1$, then $C_k^S(\ket{\psi})$ becomes:
\begin{equation}
    \frac{1}{|\mathcal{G}_k|}\sum_{\vec{\mathbf{k}} \vdash k} N_{\mathcal{G}}(\vec{\mathbf{k}}) = 1.
\end{equation}
Then $C_k^S(\ket{\psi}) = 1$ if and only if $\rho_S$ is pure, which for a global pure state $\ket{\psi}$ implies separability across $S|S^c$.
\end{proof}
\begin{lemma}
    $\mathcal{E}_k^S(\ket{\psi})$ cannot increase under local operations and classical communication (LOCC) for $\mathcal{G}=\mathcal{S}, \mathcal{C},\mathcal{D}$, both deterministically and on average. 
    \label{lemma:locc}
\end{lemma}
We separate the proof of Lemma~\ref{lemma:locc} into two steps. 
In the {\em first step}, we begin by showing that, for any single pure-state outcome obtained via LOCC deterministically, the Symmetrized Entanglement does not exceed that of the initial pure state.

Let the Schmidt decompositions of two bipartite pure states $\ket{\psi}$ and $\ket{\phi}$ be:
\begin{equation}
\begin{split}
    \ket{\psi} &= \sum_{j=1}^{r}\sqrt{{\lambda_{\psi}}_j}  \ket{\psi'_j}_S  \otimes \ket{\psi''_j}_{S^c}, \\
    \ket{\phi} &= \sum_{j=1}^{r}\sqrt{{\lambda_{\phi}}_j}  \ket{\phi'_j}_S  \otimes \ket{\phi''_j}_{S^c},
\end{split}
\end{equation}
with reduced states:
\begin{equation}
\begin{split}
    {\rho_{\psi}}_S=\sum_{j=1}^{r}{\lambda_{\psi}}_j\ket{\psi'_j}\bra{\psi'_j}, \\
    {\rho_{\phi}}_S=\sum_{j=1}^{r}{\lambda_{\phi}}_j\ket{\phi'_j}\bra{\phi'_j}.
\end{split}
\end{equation}
By Nielsen's majorization theorem~\cite{nielsen1999conditions}, the deterministic LOCC transformation $\ket{\psi} \xrightarrow{\mathrm{LOCC}} \ket{\phi}$ if and only if the Schmidt vectors satisfy $\boldsymbol{\lambda}_{\psi} \prec \boldsymbol{\lambda}_{\phi}$.  
Since the functions $\{x_i\} \mapsto \sum_i x_i^l$ for $l \geqslant 1$ are Schur-convex~\cite{bhatia1997matrix, marshall2011inequalities}, we have:
\begin{equation}
\tr(\rho_{\psi,S}^l)=\sum_{j=1}^{r}{\lambda_{\psi}}_j^l\leqslant\tr(\rho_{\phi,S}^l)=\sum_{j=1}^{r}{\lambda_{\phi}}_j^l, \quad \forall\, l \geqslant 1.
\end{equation}
Moreover, we know that $C_k^S(\ket{\psi})$ is a polynomial in $\{\tr(\rho_S^l)\}_{l \le k}$ with nonnegative coefficients $N_{\mathcal{G}}(\vec{\mathbf{k}}) / |\mathcal{G}_k|$.
This implies that:
\begin{equation}
    C_k^S(\ket{\psi})\leqslant C_k^S(\ket{\phi}),
\end{equation}
or equivalently:
\begin{equation}
    \mathcal{E}_k^S(\ket{\psi})\geqslant\mathcal{E}_k^S(\ket{\phi}),
\end{equation}
which shows the pure-state LOCC monotonicity of $\mathcal{E}_k^S(\ket{\psi})$.

In the {\em second step}, we show that $\mathcal{E}_k^S(\ket{\psi})$ for $\mathcal{G}=\mathcal{S},\mathcal{C},\mathcal{D}$ cannot increase \emph{on average} under local operations and classical communication (LOCC),
since a LOCC applied to $\ket{\psi}$ may also produce multiple pure-state outcomes ${\ket{\phi_t}}$ with associated probabilities ${p_t}$, yielding a mixed state, i.e., $\ket{\psi} \xrightarrow{\mathrm{LOCC}} \sum_{t}p_t\ket{\phi_t}\bra{\phi_t}$.
We start by proving the following lemma:
\begin{lemma}
    $\mathcal{E}_k^S(\ket{\psi})$ is both concave and Schur-concave in terms of the spectra of $\rho_S$ for $\mathcal{G}=\mathcal{S}$.
\label{lemma:convex}
\end{lemma}
\begin{proof}
We start from the proof for $\mathcal{G}=\mathcal{S}$. Since $a_{\mathcal{S}_k}=1$ (proved in Supplemental Material~\ref{sm:exemplary-cases}):
\begin{equation}
    C_k^S(\ket{\psi}, \mathcal{S}) = \sum_{\substack{g_1+\cdots+g_{r}=k \\ g_i\in\mathbb{Z} \ \mathrm{and} \ 0\leqslant g_i\leqslant k}} \lambda_1^{g_1}\lambda_2^{g_2} \cdots \lambda_r^{g_r}
\end{equation}
where $\{\lambda_i\}_{i=1}^r$ are the eigenvalues of $\rho_S$ and $r$ is the rank of $\rho_S$. 
To prove the claim, we will first see that the Symmetrized Entanglement is given by the expected value of a polynomial function of an exponential distribution of random variables, as follows.

Let $Z_1,\cdots,Z_r$ be random variables following the independent and identically distributed standard exponential distribution respectively, i.e., the probability density function $f(Z_i)$ is:
\begin{equation}
    f(Z_i)=
    \begin{cases}
    e^{-Z_i} & (Z_i \geqslant 0), \\
    0 & (Z_i < 0).    
    \end{cases}
\end{equation}
Notably, $\mathbb{E}\left[ {Z_i}^m \right]=m!$ for $m\in\mathbb{N}$.
Consider the polynomial:
\begin{equation}
    \left( \sum_{i=1}^{r}\lambda_i Z_i \right)^k = \sum_{\substack{g_1+\cdots+g_{r}=k \\ g_i\in\mathbb{Z} \ \mathrm{and} \ 0\leqslant g_i\leqslant k}} \frac{k!}{g_1! \cdots g_r!}\lambda_1^{g_1}\cdots\lambda_r^{g_r}Z_1^{g_1}\cdots Z_r^{g_r}.
\end{equation}
Therefore, its expectation value becomes:
\begin{equation}
    \mathbb{E}\left[ \left( \sum_{i=1}^{r}\lambda_iZ_i \right)^k \right] = \sum_{\substack{g_1+\cdots+g_{r}=k \\ g_i\in\mathbb{Z} \ \mathrm{and} \ 0\leqslant g_i\leqslant k}} \frac{k!}{g_1! \cdots g_r!}\lambda_1^{g_1}\cdots\lambda_r^{g_r}\mathbb{E}\left[Z_1^{g_1}\right]\cdots \mathbb{E}\left[Z_r^{g_r}\right]=k! \sum_{\substack{g_1+\cdots+g_{r}=k \\ g_i\in\mathbb{Z} \ \mathrm{and} \ 0\leqslant g_i\leqslant k}} \lambda_1^{g_1} \cdots \lambda_r^{g_r}.
\end{equation}
Therefore we have related the Symmetrized Entanglement to an exponential distribution of random variables, up to a factor of $k!$. Now we will use this relationship to show convexity.
Consider two different pure states, $\ket{\psi'}$ and $\ket{\psi''}$, we denote the spectra of their reduced states with respect to the subsystem $S$ as $\{\lambda'_i\}_{i=1}^{r}$ and $\{\lambda''_i\}_{i=1}^{r}$, respectively.
Then, for $\nu,\mu\geqslant0$, $\nu+\mu=1$ and using the convexity of the function $x^k \ (k\geqslant1)$:
\begin{equation}
\begin{split}
&\sum_{\substack{g_1+\cdots+g_{r}=k \\ g_i\in\mathbb{Z} \ \mathrm{and} \ 0\leqslant g_i\leqslant k}} (\nu\lambda'_1+\mu\lambda''_1)^{g_1} \cdots (\nu\lambda'_r+\mu\lambda''_r)^{g_r} =\frac{1}{k!}\mathbb{E}\left[ \left( \sum_{i=1}^{r}(\nu\lambda'_i+\mu\lambda''_i)Z_i \right)^k \right] \\
=&\frac{1}{k!}\mathbb{E}\left[ \left( \nu\sum_{i=1}^{r}\lambda'_iZ_i + \mu\sum_{i=1}^{r}\lambda''_iZ_i \right)^k \right] \leqslant \frac{1}{k!} \mathbb{E}\left[ \nu \left(\sum_{i=1}^{r}\lambda'_iZ_i\right)^k + \mu \left(\sum_{i=1}^{r}\lambda''_iZ_i\right)^k \right] \\
=&\nu\sum_{\substack{g_1+\cdots+g_{r}=k \\ g_i\in\mathbb{Z} \ \mathrm{and} \ 0\leqslant g_i\leqslant k}} {\lambda_1'}^{g_1} \cdots {\lambda_r'}^{g_r} + \mu\sum_{\substack{g_1+\cdots+g_{r}=k \\ g_i\in\mathbb{Z} \ \mathrm{and} \ 0\leqslant g_i\leqslant k}} {\lambda_1''}^{g_1} \cdots {\lambda_r''}^{g_r}.
\end{split}
\end{equation}
Therefore, $C_k^S(\ket{\psi}, \mathcal{S})$ ($\mathcal{E}_k^S(\ket{\psi}, \mathcal{S})$) is convex (concave) in the spectra of $\rho_S$.
Note that the similar statement can also be found in~\cite{sra2020new}.

To show Schur-convexity of $C_k^S(\ket{\psi}, \mathcal{S})$, we denote $\tau_{l}:=\sum_{i=1}^{r}\lambda_i^l =\tr(\rho_S^l)$ and we define $C_0^S(\ket{\psi})=C_1^S(\ket{\psi})=1$. With this we define and develop the following  generating function:
\begin{equation}
    G^{(\mathcal{S})}(x, \boldsymbol{\lambda})=\exp\left( \sum_{l=1}^{\infty} \frac{\tau_l}{l}x^l \right) = \prod_{l=1}^{\infty}\sum_{m_l=0}^{\infty}\frac{\tau_l^{m_l}}{l^{m_l}m_l!}x^{lm_l} = \sum_{k=0}^{\infty}x^k\sum_{\vec{\mathbf{k}} \vdash k} \prod_{l} \frac{\tau_{l}^{m_l}}{l^{m_l}m_l!} = \sum_{k=0}^{\infty}C_k^S(\ket{\psi}, \mathcal{S})x^k.
\label{eq:G-func-symm}
\end{equation} 

The final equality is due to Eq.~\eqref{eq:general-partition} where $|\mathcal{S}_k| = k!$ and $N_{\mathcal{S}}(\vec{\mathbf{k}}) = k! / (\prod_l l^{m_l} m_l!)$ for $\sum_{l}lm_l=k$.
Also, since $\log (1-x) = - \sum_{l=1}^{\infty} (x^l / l)$, we also have:
\begin{equation}
\begin{split}
    G^{\mathcal{(S)}}(x, \boldsymbol{\lambda})=&\exp\left( \sum_{l=1}^{\infty} \frac{\tau_l}{l}x^l \right) = \exp\left( \sum_{l=1}^{\infty}\sum_{i=1}^{r}\frac{\lambda_i^l}{l}x^l \right) = \prod_{i=1}^{r}\exp\left( \sum_{l=1}^{\infty} \frac{\lambda_i^l}{l}x^l \right) \\
    =& \prod_{i=1}^{r}\exp\left( -\log(1-\lambda_i x)  \right) = \prod_{i=1}^{r}\frac{1}{1-\lambda_i x}.
\end{split}
\end{equation}
To see Schur-convexity we will follow the Schur-Ostrowski criterion~\cite{peajcariaac1992convex}. Namely we want to show that:
\begin{equation}
    (\lambda_i - \lambda_j)\left( \frac{\partial C_k^S(\ket{\psi}, \mathcal{S})}{\partial \lambda_i} - \frac{\partial C_k^S(\ket{\psi}, \mathcal{S})}{\partial \lambda_j} \right)  \geqslant 0.
\end{equation}
The partial derivative of $G^{(\mathcal{S})}(x, \boldsymbol{\lambda})$ with respect to certain $\lambda_i$ and $1 \leqslant i \leqslant r$ is given by
\begin{equation}
\frac{\partial G^{(\mathcal{S})}(x, \boldsymbol{\lambda})}{\partial \lambda_i} = \sum_{k=0}^{\infty}\frac{\partial C_k^S(\ket{\psi}, \mathcal{S})}{\partial \lambda_i}x^k = \frac{x}{1-\lambda_i x} G^{(\mathcal{S})}(x, \boldsymbol{\lambda}) = \left( \sum_{a=1}^{\infty} \lambda_i^{a-1} x^{a} \right) \left( \sum_{k=0}^{\infty} C_{k}^S(\ket{\psi}, \mathcal{S})x^k \right).
\end{equation}
We collect the terms with the same powers on $x$, and arrive at
\begin{equation}
    \frac{\partial C_k^S(\ket{\psi}, \mathcal{S})}{\partial \lambda_i} = \sum_{m=0}^{k-1}\lambda_i^{k-m-1} C_m^S(\ket{\psi}).
\end{equation}
Therefore, for $i\neq j$ we have:
\begin{equation}
    (\lambda_i - \lambda_j)\left( \frac{\partial C_k^S(\ket{\psi}, \mathcal{S})}{\partial \lambda_i} - \frac{\partial C_k^S(\ket{\psi}, \mathcal{S})}{\partial \lambda_j} \right) = \sum_{m=0}^{k-1}C_m^S(\ket{\psi}, \mathcal{S})(\lambda_i - \lambda_j) (\lambda_i^{k-m-1} - \lambda_j^{k-m-1}) \geqslant 0.
\end{equation}
Therefore, $C_k^S(\ket{\psi}, \mathcal{S})$ ($\mathcal{E}_k^S(\ket{\psi}, \mathcal{S})$) is also Schur-convex (Schur-concave) in the spectra of $\rho_S$.
\end{proof}
Finally, we consider the outcomes by applying one round of LOCC on either system $S$ or $S^c$. If the LOCC is applied on system $S^c$, let $\{M_t\}$ be the Kraus operators on $S^c$ with $\sum_{t}M_t^{\dagger}M_t=\mathbb{I}^{S^c}$. 
Therefore:
\begin{equation}
    \ket{\phi_t}=\frac{(\mathbb{I}^S\otimes M_t)}{\sqrt{p_t}}\ket{\psi},
\end{equation}
and $p_t=\bra{\psi}(\mathbb{I}^S\otimes M_t^{\dagger} M_t)\ket{\psi}$. 
Let $\rho_{S,t}=\tr_{S^c}(\ket{\phi_t}\bra{\phi_t})$ and $\rho_S=\tr_{S^c}(\ket{\psi}\bra{\psi})$, therefore:
\begin{equation}
    \sum_{t}p_t\rho_{S,t} = \tr_{S^c}\left( \sum_{t}(\mathbb{I}^S\otimes M_t)\ket{\psi}\bra{\psi}(\mathbb{I}^S\otimes M^{\dagger}_t) \right) = \rho_S.
\end{equation}
Note that for $\mathcal{G}=\mathcal{S}$, from Lidskii - Ky Fan inequality~\cite{bhatia1997matrix}, for any two $d\times d$ Hermitian matrices $A$ and $B$ and any $1\leqslant v\leqslant d$ we have:
\begin{equation}
\lambda_1(A+B) + \cdots + \lambda_{v}(A+B) \leqslant \lambda_1(A) + \cdots + \lambda_{v}(A) + \lambda_1(B) + \cdots + \lambda_{v}(B),
\end{equation}
and we force $\lambda_1 \geqslant\cdots\geqslant\lambda_v$. 
Then we group up the terms on the right-hand-side such that $(\lambda_x(A)+\lambda_y(B))$ for $1\leqslant x, y\leqslant d$ are also placed in a non-increasing order, therefore:
\begin{equation}
\begin{split}
\lambda_1(A+B) + \cdots + \lambda_{v}(A+B) &\leqslant \lambda_1(A) + \cdots + \lambda_{v}(A) + \lambda_1(B) + \cdots + \lambda_{v}(B) \\
&\leqslant (\lambda_{x_1}(A)+\lambda_{y_1}(B))+\cdots+(\lambda_{x_v}(A)+\lambda_{y_v}(B)).    
\end{split}
\end{equation}
Therefore:
\begin{equation}
\boldsymbol{\lambda}(A+B) \prec \boldsymbol{\lambda}(A) + \boldsymbol{\lambda}(B).
\end{equation}
Then:
\begin{equation}
\boldsymbol{\lambda}\left(\sum_{t}p_t\rho_{S,t}\right) \prec \sum_{t}p_t\boldsymbol{\lambda}(\rho_{S,t}).
\end{equation}
Note that this result has also been illustrated in~\cite{nielsen2001majorization} (cf. Theorem 11). 
By using the Schur-concavity and concavity of $\mathcal{E}_k^S(\ket{\psi}, \mathcal{S})$ respectively, we have:
\begin{equation}
\begin{split}
&\mathcal{E}_k^S\left(\ket{\psi}, \mathcal{S}\right) = \mathcal{E}_k^S\left(\rho_S, \mathcal{S}\right) = \mathcal{E}_k^S\left( \boldsymbol{\lambda}\left(\sum_{t}p_t\rho_{S,t}\right), \mathcal{S}\right)\\
\geqslant& \mathcal{E}_k^S\left(\sum_{t}p_t\boldsymbol{\lambda}(\rho_{S,t}), \mathcal{S}\right) \geqslant \sum_{t} p_t \mathcal{E}_k^S\left( \boldsymbol{\lambda}(\rho_{S,t}), \mathcal{S} \right) = \sum_{t} p_t \mathcal{E}_k^S\left( \ket{\phi_t}, \mathcal{S} \right),
\end{split}
\end{equation}
as desired for $\mathcal{G}=\mathcal{S}$. 
The notation $\mathcal{E}_k^S(\boldsymbol{\lambda})$ means that $\mathcal{E}_k^S$ is computed directly from the spectra $\boldsymbol{\lambda}$. Thus we have shown that $\mathcal{E}_k^S$ is monotonic under LOCC on average for the full symmetric group $\mathcal{S}$. 
For $\mathcal{G}=\mathcal{C}$ and $\mathcal{D}$, since the trace polynomials like $\left[\tr(\rho^a)\right]^b$ are convex in terms of the state $\rho$ for $a\geqslant 1$ and $ab \geqslant 1$~\cite{hu2006generalized}.
From Supplemental Material~\ref{sm:exemplary-cases}, we can see that both $\mathcal{C}_k^S\left(\ket{\psi}, \mathcal{C}\right)$ and $\mathcal{C}_k^S\left(\ket{\psi}, \mathcal{D}\right)$ are a linear combination of convex trace polynomials with nonnegative coefficients.
Therefore, $\mathcal{E}_k^S\left(\ket{\psi}, \mathcal{C}\right)$ and $\mathcal{E}_k^S\left(\ket{\psi}, \mathcal{D}\right)$ are also concave in terms of the state $\rho$. 
Therefore for $\mathcal{G}=\mathcal{C}$:
\begin{equation}
\mathcal{E}_k^S\left(\ket{\psi}, \mathcal{C}\right) = \mathcal{E}_k^S\left(\rho_S, \mathcal{C}\right) = \mathcal{E}_k^S\left(\sum_{t}p_t\rho_{S,t}, \mathcal{C}\right) \geqslant \sum_{t}p_t\mathcal{E}_k^S\left( \rho_{S,t},\mathcal{C} \right) = \sum_{t}p_t\mathcal{E}_k^S\left( \ket{\phi_t},\mathcal{C} \right).
\end{equation}
\begin{equation}
\mathcal{E}_k^S\left(\ket{\psi}, \mathcal{D}\right) = \mathcal{E}_k^S\left(\rho_S, \mathcal{D}\right) = \mathcal{E}_k^S\left(\sum_{t}p_t\rho_{S,t}, \mathcal{D}\right) \geqslant \sum_{t}p_t\mathcal{E}_k^S\left( \rho_{S,t},\mathcal{D} \right) = \sum_{t}p_t\mathcal{E}_k^S\left( \ket{\phi_t},\mathcal{D} \right).
\end{equation}
On the other hand, if the LOCC is applied on system $S$, let $\{K_t\}$ be the Kraus operators on $S$ with $\sum_{t}K_t^{\dagger}K_t=\mathbb{I}^{S}$.
Therefore:
\begin{equation}
    \ket{\phi_t}=\frac{(K_t \otimes \mathbb{I}^{S^c})}{\sqrt{p_t}}\ket{\psi},
\end{equation}
and $p_t=\bra{\psi}(K_t^{\dagger} K_t \otimes \mathbb{I}^{S^c})\ket{\psi}$. 
Let $\rho_{S^c,t}=\tr_{S}(\ket{\phi_t}\bra{\phi_t})$ and $\rho_{S^c}=\tr_{S}(\ket{\psi}\bra{\psi})$, therefore:
\begin{equation}
    \sum_{t}p_t\rho_{S^c,t}=\tr_S\left(\sum_{t}(K_t\otimes \mathbb{I}^{S^c})\ket{\psi}\bra{\psi}(K^{\dagger}_t\otimes \mathbb{I}^{S^c})\right) = \rho_{S^c}.
\end{equation}
Note that $\mathcal{E}_k^S = \mathcal{E}_k^{S^c}$ since $\tr(\rho_S^l) = \tr(\rho_{S^c}^l)$.
Therefore:
\begin{equation}
\begin{split}
&\mathcal{E}_k^S\left(\ket{\psi}, \mathcal{S}\right) = \mathcal{E}_k^{S^c}\left(\ket{\psi}, \mathcal{S}\right) = \mathcal{E}_k^{S^c}\left(\rho_{S^c}, \mathcal{S}\right) = \mathcal{E}_k^{S^c}\left(\boldsymbol{\lambda}\left(\sum_{t}p_t\rho_{{S^c},t}\right), \mathcal{S}\right)\\
\geqslant& \mathcal{E}_k^{S^c}\left(\sum_{t}p_t\boldsymbol{\lambda}(\rho_{{S^c},t}), \mathcal{S}\right) \geqslant \sum_{t} p_t \mathcal{E}_k^{S^c}\left( \boldsymbol{\lambda}(\rho_{{S^c},t}), \mathcal{S} \right) = \sum_{t} p_t \mathcal{E}_k^{S^c}\left( \ket{\phi_t}, \mathcal{S} \right) = \sum_{t} p_t \mathcal{E}_k^{S}\left( \ket{\phi_t}, \mathcal{S} \right),
\end{split}
\end{equation}
and we can apply similar approach for $\mathcal{G}=\mathcal{C},\mathcal{D}$. 
This shows that $\mathcal{E}_k^S(\ket{\psi})$ cannot increase under LOCC for $\mathcal{G}=\mathcal{S}, \mathcal{C},\mathcal{D}$, both deterministically and on average, which proves Lemma~\ref{lemma:locc}. 
\begin{lemma}
    $\mathcal{E}_k^S(\ket{\psi})$ is continuous on the set of pure states for any permutation group projector $P_k^S$. Specifically, if $||\ket{\psi_1}\bra{\psi_1}-\ket{\psi_2}\bra{\psi_2}||_1\leqslant\varepsilon$, then $|\mathcal{E}_k^S(\ket{\psi_1}) - \mathcal{E}_k^S(\ket{\psi_2})|\leqslant \sqrt{k}\varepsilon$.
\label{lemma:continuous}
\end{lemma}
\begin{proof}
We first recall that
\begin{equation}
    C_k^S(\ket{\psi}) := \bra{\psi}^{\otimes k} \left( P_k^S \otimes \mathbb{I}_k^{S^c} \right) \ket{\psi}^{\otimes k} = \tr\left( \left( P_k^S \otimes \mathbb{I}_k^{S^c} \right) \rho^{\otimes k} \right). 
\end{equation}
Then, using Hölder’s inequality for Schatten norms:
\begin{equation}
    |C_k^S(\ket{\psi_1}) - C_k^S(\ket{\psi_2})| = \left|\tr\left( \left( P_k^S \otimes \mathbb{I}_k^{S^c} \right) (\rho_1^{\otimes k} - \rho_2^{\otimes k}) \right) \right|\leqslant|| P_k^S \otimes \mathbb{I}_k^{S^c} ||_{\infty} || \rho_1^{\otimes k} - \rho_2^{\otimes k} ||_{1} = || \rho_1^{\otimes k} - \rho_2^{\otimes k} ||_{1},
\end{equation}
as $P_k^S$ is a projector with eigenvalues of 0 and 1 only and therefore the largest eigenvalue, which is equal to the infinite norm, is 1. 
Here we denote $\rho_1 = \ket{\psi_1}\bra{\psi_1}$ and $\rho_2 = \ket{\psi_2}\bra{\psi_2}$.
As $\rho_1^{\otimes k}$ and $\rho_2^{\otimes k}$ are both pure, then:
\begin{equation}
    || \rho_1^{\otimes k} - \rho_2^{\otimes k} ||_{1} = 2\sqrt{1 - | \langle\psi_1|\psi_2\rangle |^{2k}}.
\end{equation}
Note that:
\begin{equation}
    ||\ket{\psi_1}\bra{\psi_1}-\ket{\psi_2}\bra{\psi_2}||_1 = 2\sqrt{1 - | \langle\psi_1|\psi_2\rangle |^{2} }\leqslant\varepsilon
\end{equation}
Then:
\begin{equation}
    | \langle\psi_1|\psi_2\rangle |^{2}\geqslant1-\frac{\varepsilon^2}{4}.
\end{equation}
Therefore:
\begin{equation}
    |\mathcal{E}_k^S(\ket{\psi_1}) - \mathcal{E}_k^S(\ket{\psi_2})| = |C_k^S(\ket{\psi_1}) - C_k^S(\ket{\psi_2})| \leqslant || \rho_1^{\otimes k} - \rho_2^{\otimes k} ||_{1} \leqslant 2\sqrt{1 - (1 - \varepsilon^2 / 4)^k}\leqslant 2\sqrt{1 - (1 - k\varepsilon^2 / 4)} = \sqrt{k}\varepsilon.
\end{equation}
\end{proof}
So far we have shown that $\mathcal{E}_k^S(\ket{\psi})$ is a well-defined bipartite pure-state entanglement measure for $\mathcal{G}=\mathcal{S},\mathcal{C},\mathcal{D}$.
We are then able to generalize this measure to the multipartite setting by taking the average over all bipartitions satisfying $|S| = s$, in line with the method of~\cite{scott2004multipartite}. 
In addition, genuinely multipartite entanglement can also be quantified by the smallest entanglement among all bipartitions of the state in hand $\ket{\psi}$, as demonstrated in~\cite{ma2011measure}.
Theorem~\ref{thrm:general-ent} therefore introduces a family of well-defined measure of both bipartite and multipartite entanglement.

Note that to extend this measure from pure states to general mixed states, one needs to apply the convex roof optimization~\cite{vidal2000entanglement}:
\begin{equation}
    \boldsymbol{\mathcal{E}}_{k}^{S}(\rho) = \inf_{\{p_t,\ket{\psi_t}\}}\sum_{t}p_t\,\mathcal{E}_{k}^{S}(\ket{\psi_t}),
\end{equation}
where the infimum is taken over all possible decompositions $\rho=\sum_{t}p_t\ket{\psi_t}\bra{\psi_t}$. 
We use the boldface notation $\boldsymbol{\mathcal{E}}_{k}^{S}(\rho)$ to distinguish it from $\mathcal{E}_{k}^{S}(\rho)$, where the former denotes the entanglement measure for an arbitrary state $\rho$, while the latter refers to the corresponding computable quantity in terms of the reduced state of $\rho$ or its spectra $\boldsymbol{\lambda}(\rho)$.

\section{Representing $C_k^S(\ket{\psi})$ with State Moments and Reduced State Spectra~\label{sm:part-moments-spectra}}
Using the coordinate-free definition of the partial trace, 
\begin{equation}
    \tr\left( (M \otimes \mathbb{I}) \rho \right) = \tr\left( M \, \rho_A \right),
\end{equation}
where $\rho_A = \tr_B(\rho)$ is the reduced state on system $A$, 
we have:
\begin{equation}
    C_k^S(\ket{\psi}) 
    = \bra{\psi}^{\otimes k} \left( P_k^S \otimes \mathbb{I}_k^{S^c} \right) \ket{\psi}^{\otimes k} 
    = \tr\left( \left( P_k^S \otimes \mathbb{I}_k^{S^c} \right) \rho^{\otimes k} \right) 
    = \tr \left( P_k^S \rho_S^{\otimes k} \right).
\end{equation}
Extending the SWAP trick to arbitrary permutations gives~\cite{barenco1997stabilization, buhrman2001quantum, koczor2021exponential, quek2024multivariate, liu2025generalized}:
\begin{equation}
    \bra{\psi}^{\otimes k}\pi_S\ket{\psi}^{\otimes k} 
    = \tr(\pi_S \rho^{\otimes k}) 
    = \prod_{l} \left[ \tr(\rho_S^l) \right]^{m_l},
\end{equation}
where $\pi$ has cycle type $1^{m_1} 2^{m_2} \cdots t^{m_t}$, i.e., $m_l$ cycles of length $l$ with $\sum_{l} l\,m_l = k$.  
For a partition of $k$ with fixed $l$ and $m_l$, we denote the number of permutations $\pi \in \mathcal{G}_k$ with this cycle type as $N_{\mathcal{G}}(\vec{\mathbf{k}})$.
Therefore:
\begin{equation}
\begin{split}
    C_k^S(\ket{\psi}) = \frac{1}{|\mathcal{G}_k|}\sum_{\vec{\mathbf{k}} \vdash k} N_{\mathcal{G}}(\vec{\mathbf{k}}) \prod_{l}\left[\tr(\rho_S^l)\right]^{m_l}.
\end{split}
\end{equation}

The quantity $C_k^S(\ket{\psi})$ can also be expressed in terms of partitioned spectra of $\rho_S$.
We hereby illustrate the detailed derivations.
Since $\tr(\rho_S^l)=\sum_{i=1}^{r}\lambda_i^l$, then we have: 
\begin{equation}
\begin{split}
    C_k^S(\ket{\psi}) =& \frac{1}{|\mathcal{G}_k|}\sum_{\vec{\mathbf{k}} \vdash k} N_{\mathcal{G}}(\vec{\mathbf{k}})\prod_{l=1}^{t} \left(\sum_{i=1}^{r}\lambda_i^l\right)^{m_l} \\
    =& \frac{1}{|\mathcal{G}_k|}\sum_{\vec{\mathbf{k}} \vdash k} N_{\mathcal{G}}(\vec{\mathbf{k}})\prod_{l=1}^{t} \left( \sum_{  \substack{\sum_{i}c_{i,l}=m_l\\
     c_{i,l}\in\mathbb{Z}_{\geqslant0}} } \frac{m_l!}{c_{1,l}!\cdots c_{r,l}!} \prod_{i=1}^{r}\lambda_i^{lc_{i,l}}  \right) \\
    =& \frac{1}{|\mathcal{G}_k|}\sum_{\vec{\mathbf{k}} \vdash k} N_{\mathcal{G}}(\vec{\mathbf{k}}) \sum_{\substack{\sum_{i}c_{i,1}=m_1 \\ c_{i,1}\in\mathbb{Z}_{\geqslant0}} }\cdots\sum_{\substack{\sum_{i}c_{i,t}=m_t \\ c_{i,t}\in\mathbb{Z}_{\geqslant0}} } \left( \prod_{l=1}^{t} \frac{m_l}{c_{1,l}!\cdots c_{r,l}!} \right) \left( \prod_{l=1}^{t} \prod_{i=1}^{r} \lambda_i^{lc_{i,l}} \right) \\
    =& \frac{1}{|\mathcal{G}_k|}\sum_{\vec{\mathbf{k}} \vdash k} N_{\mathcal{G}}(\vec{\mathbf{k}}) \sum_{\substack{\sum_{i}c_{i,x}=m_x, \forall x\in\{1,\cdots,t\}  \\ c_{i,x}\in\mathbb{Z}_{\geqslant0}}} \left( \prod_{l=1}^{t} \frac{m_l}{c_{1,l}!\cdots c_{r,l}!} \right) \left( \prod_{i=1}^{r}\lambda_i^{\sum_{l=1}^{t}lc_{i,l}} \right) 
\end{split}
\end{equation}
We group up the terms such that $\sum_{l}lc_{i,l}=g_i$, then:
\begin{equation}
\begin{split}
C_k^S(\ket{\psi})  =& \frac{1}{|\mathcal{G}_k|}\sum_{\vec{\mathbf{k}} \vdash k} N_{\mathcal{G}}(\vec{\mathbf{k}}) \sum_{\substack{\sum_{i}c_{i,x}=m_x, \forall x\in\{1,\cdots,t\} \\ c_{i,x}\in\mathbb{Z}_{\geqslant0}}} \left( \prod_{l=1}^{t} \frac{m_l}{c_{1,l}!\cdots c_{r,l}!} \right) \lambda_1^{g_1}\cdots\lambda_r^{g_r} \\
    =& \sum_{\vec{\mathbf{k}} \vdash k} a_{\mathcal{G}_k}(\vec{\mathbf{k}})\lambda_1^{g_1}\cdots\lambda_r^{g_r},
\end{split}
\end{equation}
where:
\begin{equation}
    a_{\mathcal{G}_k}(\vec{\mathbf{k}}) 
    = \frac{N_{\mathcal{G}}(\vec{\mathbf{k}})}{|\mathcal{G}_k|} \sum_{\substack{\sum_{i}c_{i,x}=m_x, \forall x\in\{1,\cdots,t\} \\ c_{i,x}\in\mathbb{Z}_{\geqslant0}}} \left( \prod_{l=1}^{t} \frac{m_l}{c_{1,l}!\cdots c_{r,l}!} \right),
\end{equation}
and notably, $\sum_{i}g_i=\sum_{i}\sum_{l}lc_{i,l}=\sum_{l}lm_l=k$. 

\section{Simple Expressions of $C_k^S(\ket{\psi})$ under groups $\mathcal{S},\mathcal{C},\mathcal{D}$~\label{sm:exemplary-cases}}
The similar results have been illustrated in~\cite{bradshaw2023cycle}.
Here we revisit the results and list the proofs. 
\begin{enumerate}
\item We start from the case of $\mathcal{G}=\mathcal{S}$.
Notably:
\begin{equation}
N_{\mathcal{S}}(\vec{\mathbf{k}}) = |\{ \pi\in\mathcal{S}_k: \text{cycle type } \vec{\mathbf{k}}\}|= \frac{k!}{\prod_{l}l^{m_l}m_l!}.
\end{equation}
The intuition is as follows. 
We start with $k!$ possible arrangements of the $k$ labeled elements.  
Each $l$-cycle has $l$ rotational symmetries, meaning the starting point of the cycle can be chosen arbitrarily.  
Thus, we divide by $l$ for each such cycle, giving a factor of $\prod_{l} l^{m_l}$ in the denominator.  
Finally, as cycles of the same length are indistinguishable, we therefore divide by $m_l!$ for each $l$, resulting in an additional factor of $\prod_{l} m_l!$ in the denominator.

Therefore, due to Eq.~\eqref{eq:general-partition}:
\begin{equation}
    C_k^S(\ket{\psi}, \mathcal{S}) = \sum_{\vec{\mathbf{k}} \vdash k} \prod_{l} \frac{\left[ \tr(\rho_S^l) \right]^{m_l}}{l^{m_l} m_l!}.
\label{eq:symmetric-part}
\end{equation}
From Eq.~\eqref{eq:G-func-symm}, we have:
\begin{equation}
    G^{(\mathcal{S})}(x, \boldsymbol{\lambda})=\exp\left( \sum_{l=1}^{\infty} \frac{\tau_l}{l}x^l \right)  = \sum_{k=0}^{\infty}C_k^S(\ket{\psi}, \mathcal{S})x^k
\end{equation}
By taking the partial derivative of $G^{(\mathcal{S})}(x, \boldsymbol{\lambda})$ on $x$, we have:
\begin{equation}
\frac{\partial G^{(\mathcal{S})}(x, \boldsymbol{\lambda})}{\partial x} = \sum_{k=0}^{\infty}(k+1) C_{k+1}^S(\ket{\psi}, \mathcal{S}) x^k,
\end{equation}
and:
\begin{equation}
\begin{split}
    \frac{\partial G^{(\mathcal{S})}(x, \boldsymbol{\lambda})}{\partial x} =& \exp\left( \sum_{l=1}^{\infty} \frac{\tau_l}{l}x^l \right) \cdot \sum_{l=1}^{\infty}\tau_l x^{l-1} = \sum_{k=0}^{\infty}C_k^S(\ket{\psi}, \mathcal{S})x^k \cdot \sum_{l=0}^{\infty}\tau_{l+1}x^l \\
    =& \sum_{k=0}^{\infty}\left( \sum_{l=0}^{\infty}C_k^S(\ket{\psi}, \mathcal{S})\tau_{l+1}  \right) x^{k+l} = \sum_{k=0}^{\infty} \left(  \sum_{q=0}^{k}C_q^S(\ket{\psi}, \mathcal{S}) \tau_{k-q+1} \right) x^k.
\end{split}
\end{equation}
In order to match the coefficients of $x^k$, we have:
\begin{equation}
    k C_{k}^S(\ket{\psi}, \mathcal{S}) = \sum_{q=0}^{k-1}C_q^S(\ket{\psi}, \mathcal{S}) \tau_{k-q}.
\end{equation}
Thus:
\begin{equation}
    C_{k}^S(\ket{\psi}, \mathcal{S}) = \frac{1}{k} \sum_{q=0}^{k-1}C_q^S(\ket{\psi}, \mathcal{S}) \tau_{k-q} = \frac{1}{k}\sum_{q=0}^{k-1}C_q^S(\ket{\psi}, \mathcal{S}) \tr(\rho_S^{k-q}).
\label{eq:recur}
\end{equation}
Then, by using:
\begin{equation}
\log(1-x) = -\sum_{l=1}^{\infty} \frac{x^l}{l},
\end{equation}
we have:
\begin{equation}
    G^{(\mathcal{S})}(x, \boldsymbol{\lambda}) = \prod_{i=1}^{r} \left( \sum_{g=0}^{\infty} \lambda_i^g x^g \right) = \sum_{k=0}^{\infty} \left( \sum_{\substack{g_1+\cdots+g_{r}=k \\ g_i\in\mathbb{Z} \ \mathrm{and} \ 0\leqslant g_i\leqslant k}} \lambda_1^{g_1}\lambda_2^{g_2} \cdots \lambda_r^{g_r} \right) x^k.
\end{equation}
Therefore:
\begin{equation}
    C_k^S(\ket{\psi}, \mathcal{S}) = \sum_{\substack{g_1+\cdots+g_{r}=k \\ g_i\in\mathbb{Z} \ \mathrm{and} \ 0\leqslant g_i\leqslant k}} \lambda_1^{g_1}\lambda_2^{g_2} \cdots \lambda_r^{g_r},
\label{eq:partition}
\end{equation}
\item Secondly, we consider the case of $\mathcal{G}=\mathcal{C}$.
Recall that the cyclic group $\mathcal{C}_k$ consists of $k$ rotations of $k$ elements, generated by a single cycle $D=(1 \ 2 \ \cdots \ k)$.
Note that the permutation $D^j=(1 \ 2 \ \cdots \ k)^j$ for $j=1,\cdots,k$ can be decomposed into $\gcd(j,k)$ disjoint cycles, each of length $k / \gcd(j,k)$, thus with contributing term $\left[ \tr(\rho_S^{k/\gcd(j,k)}) \right]^{\gcd(j,k)}$.
Therefore:
\begin{equation}
C_k^S(\ket{\psi}, \mathcal{C})=\frac{1}{k}\sum_{j=1}^{k}\left[ \tr(\rho_S^{k/\gcd(j,k)}) \right]^{\gcd(j,k)},
\end{equation}
The similar statement can be found in~\cite{liu2025generalized} (cf. Supplemental Material I.A. and I.B.).
By using the fact that the number of $\gcd(j,k)=u$ is $\varphi(k/u)$, then:
\begin{equation}
C_k^S(\ket{\psi}, \mathcal{C})=\frac{1}{k}\sum_{u|k}\varphi\left(\frac{k}{u}\right)\left[ \tr(\rho_S^{k/u}) \right]^{u}
\end{equation}
If we denote $q=k/u$, we then have:
\begin{equation}
C_k^S(\ket{\psi}, \mathcal{C})=\frac{1}{k}\sum_{q|k}\varphi(q)\left[\tr(\rho_S^q) \right]^{\frac{k}{q}}.
\end{equation}
This is also obvious when adopting Necklace polynomials~\cite{moreau1872permutations}, a special case of Pólya enumeration theorem~\cite{redfield1927theory, polya1937kombinatorische}.
\item Finally, for $\mathcal{G}=\mathcal{D}$, we can always write $\mathcal{D}_k$ as:
\begin{equation}
\mathcal{D}_k = \{ \mathbb{I},r,r^2,\cdots,r^{k-1}, s,rs,r^2s,\cdots,r^{k-1}s \}     
\end{equation}
where $\mathbb{I}$ denotes the trivial permutation and:
\begin{equation}
r=(1 \ 2 \ \cdots \ k), \ \ s=(1 \ k)(2 \ k-1)(3 \ k-2) \cdots
\end{equation}
denote the single rotation and reflection, respectively.
Therefore, the first half is exactly $\mathcal{C}_k$ and we then focus on the second half.
Note that $r^k=\mathbb{I}$, $s^2=\mathbb{I}$ and $srs=r^{-1}$. 
We now show that $r^js= (r^js)^{-1} = sr^{-j}$.
When $j=1$, this statement is trivial. 
Suppose $r^js= (r^js)^{-1} = sr^{-j}$ holds for some $j \geqslant 1$, then for $j+1$:
\begin{equation}
r^{j+1}s= r r^js = rsr^{-j} = sr^{-(j+1)}.
\end{equation}
Therefore $r^js= (r^js)^{-1} = sr^{-j}$ holds. 

Given $k$ labels such that $x=1,2,\cdots,k$, then $r^j(x)=(x+j) \bmod{k}$ and $s(x)=k+1-x$.
Therefore, $r^j(s(x)) = (k+1-x+j) \bmod{k}$.
Since $(r^js)^2=r^js\cdot sr^{-j}=\mathbb{I}$, $r^js$ must be the cycle with only length 1 and 2, i.e., it has the cycle form of $[1^{d_1}2^{d_2}]$ and $d_1$ denotes how many fixed points there are after permuted by $r^js$.
From the involution $r^j (s(x))=x$ we have $x \equiv (k+1-x+j) \bmod{k}$.
This is equivalent to $2x \equiv (j+1) \bmod{k}$ and we need to check the number of valid solution for integer $1\leqslant x\leqslant k$.
When $k$ is odd, there is one solution for a certain $j$ that when $j$ is odd, $x=\frac{j+1}{2}$ and when $j$ is even, $x=\frac{j+k+1}{2}$.
These cases contribute the terms $\left[\tr(\rho^2)\right]^{\frac{k-1}{2}}$.
When $k$ is even, we separate the cases of odd $j$ and even $j$. 
If $j$ is even, there is no solution for $x$.
These cases contribute the terms $\left[\tr(\rho^2)\right]^{\frac{k}{2}}$.
If $j$ is odd, there are two solutions for $x$, which are $x=\frac{j+1}{2}$ and $x=\frac{j+k+1}{2}$, respectively.
These cases contribute the terms $\left[\tr(\rho^2)\right]^{\frac{k-2}{2}}$.
Therefore:
\begin{equation}
C_k^S(\ket{\psi}, \mathcal{D}) =\frac{1}{2}C_k^S(\ket{\psi}, \mathcal{C}) + 
\frac{1}{4}\left( \left[\tr(\rho_S^2)\right]^{\frac{k-2+(k \bmod 2)}{2}} + \left[\tr(\rho_S^2)\right]^{\frac{k-(k \bmod 2)}{2}} \right).
\end{equation}

\end{enumerate}
\section{Proof of Theorem~\ref{thrm:extreme}\label{sm:thrm6}}
\begin{enumerate}
\item For any state $\rho$, the quantity $\tr(\rho^l)$ attains its minimum when $\rho = \mathbb{I} / \mathbf{d}$, where $\mathbf{d}$ is the Hilbert space dimension of $\rho$.  
For example, if $\rho$ describes $s$ registers with each of local dimension $d$, then $\mathbf{d} = d^s$.
This follows from Jensen's inequality: since $x \mapsto x^l$ is convex for $l \geqslant 1$, then:
\begin{equation}
    \frac{1}{r} \tr(\rho^l) 
    = \frac{1}{r} \sum_{i=1}^r \lambda_i^l
    \geqslant \left( \frac{1}{r} \sum_{i=1}^r \lambda_i \right)^l
    = \frac{1}{r^l},
\end{equation}
where $r$ is the rank of $\rho$ and $\{\lambda_i\}_{i=1}^r$ are its eigenvalues.  
Thus:
\begin{equation}
    \tr(\rho^l) \geqslant \frac{1}{r^{l-1}} \geqslant \frac{1}{\mathbf{d}^{l-1}},
\end{equation}
and the lower bound is achieved by the maximally mixed state $\rho = \mathbb{I}/\mathbf{d}$.
Consequently:
\begin{equation}
    C_k^S(\ket{\psi},\mathcal{S}) 
    = \tr\left[ \left( P_k^S \otimes \mathbb{I}^{S^c} \right) \rho^{\otimes k} \right] 
    \geqslant \frac{1}{d^{sk}} \tr\left(\frac{1}{|\mathcal{S}_k|}\sum_{\pi\in\mathcal{S}_k}\pi_S\right) 
    = \frac{1}{d^{sk}} \binom{d^s + k - 1}{k},
\end{equation}
where the evaluation of $\tr\left(\frac{1}{|\mathcal{S}_k|}\sum_{\pi\in\mathcal{S}_k}\pi_S\right)$ has been discussed in~\cite{audenaert2006digest, nielsen2010quantum, harrow2013church}.
Then, for $\mathcal{G}=\mathcal{C}$ and $\mathcal{D}$, it is easy to show that:
\begin{equation}
C_k^S(\ket{\psi}, \mathcal{C})=\frac{1}{k}\sum_{q|k}\varphi(q)\left[\tr(\rho_S^q) \right]^{\frac{k}{q}}\geqslant\frac{1}{k}\sum_{q|k}\varphi(q)d^{\frac{ks(1-q)}{q}},
\end{equation}
and:
\begin{equation} 
\begin{split} 
C_k^S(\ket{\psi}, \mathcal{D}) &= \frac{1}{2}C_k^S(\ket{\psi}, \mathcal{C}) +  \frac{1}{4}\left( \left[\tr(\rho_S^2)\right]^{\frac{k-2+(k \bmod 2)}{2}} + \left[\tr(\rho_S^2)\right]^{\frac{k-(k \bmod 2)}{2}} \right) \\
&\geqslant \frac{1}{2k}\sum_{q|k}\varphi(q)d^{\frac{ks(1-q)}{q}} + \frac{1}{4} \left( d^{-s\frac{k-2+(k \bmod 2)}{2}} + d^{-s\frac{k-(k \bmod 2)}{2}} \right)
\end{split} 
\end{equation}
Note that from the definition, $s$-uniform pure state denotes the $n$-partite pure state such that every reduction to $s$ parties is maximally mixed, i.e., $\rho_S = \mathbb{I} / \mathbf{d}$ for $\mathbf{d}=d^s$ and any $s=|S|$, which is exactly the state that reaches maximum of $\mathcal{E}_k^s(\ket{\psi})$.
Then naturally:
\begin{equation}
\mathcal{E}_k^S(\ket{\psi}, \mathcal{S}) \ \ \text{or} \ \ \mathcal{E}_k^s(\ket{\psi}, \mathcal{S}) \leqslant 1 - \frac{1}{d^{sk}} \binom{d^s + k - 1}{k},
\end{equation}
\begin{equation}
\mathcal{E}_k^S(\ket{\psi}, \mathcal{C}) \ \ \text{or} \ \ \mathcal{E}_k^s(\ket{\psi}, \mathcal{C}) \leqslant 1 - \frac{1}{k}\sum_{q|k}\varphi(q)d^{\frac{ks(1-q)}{q}},
\end{equation}
\begin{equation}
\mathcal{E}_k^S(\ket{\psi}, \mathcal{D}) \ \ \text{or} \ \ \mathcal{E}_k^s(\ket{\psi}, \mathcal{D}) \leqslant 1 - \frac{1}{2k}\sum_{q|k}\varphi(q)d^{\frac{ks(1-q)}{q}} - \frac{1}{4} \left( d^{-s\frac{k-2+(k \bmod 2)}{2}} + d^{-s\frac{k-(k \bmod 2)}{2}} \right).
\end{equation}
\item It is obvious that $\lim_{k\rightarrow\infty}\mathcal{E}^S_k(\ket{\psi})=0$ if there is no entanglement between $S$ and $S^c$, as $\rho_S$ is pure in this case.
We will then illustrate the entangled cases.
We start with the symmetric group $\mathcal{S}$.
Denote the maximum eigenvalue of $\rho_S$ as $\lambda_{\mathrm{max}}$. Then:
\begin{equation}
0 \leqslant C_k^S(\ket{\psi}, \mathcal{S}) = \sum_{\substack{g_1+\cdots+g_{r}=k \\ g_i\in\mathbb{Z} \ \mathrm{and} \ 0\leqslant g_i\leqslant k}} \lambda_1^{g_1} \cdots \lambda_r^{g_r} \leqslant \binom{d^s+k-1}{k}\lambda^k_{\mathrm{max}}
\end{equation}
Note that for $0< \lambda_{\mathrm{max}} < 1$, we have:
\begin{equation}
0\leqslant\binom{d^s+k-1}{k}\lambda_{\mathrm{max}}^{k} = \frac{(k+1)(k+2)\cdots(k+d^s-1)}{(d^s-1)!}\lambda_{\mathrm{max}}^k\leqslant \frac{(k+d^s)^{d^s-1}}{(d^s-1)!}\lambda_{\mathrm{max}}^k.
\end{equation}
Since exponential decay beats polynomial growth under the condition that $d$ and $s$ are fixed, we have:
\begin{equation}
\lim_{k\rightarrow\infty}\frac{(k+d^s)^{d^s-1}}{(d^s-1)!}\lambda_{\mathrm{max}}^k = 0
\end{equation}
since the upper bound of the nonnegative quantity $C_k^S$ tends to zero, we generally have
\begin{equation}
\lim_{k\rightarrow\infty}C_k^S(\ket{\psi}, \mathcal{S}) = 0.
\end{equation}
Then, for the cyclic group $\mathcal{C}$, we have:
\begin{equation}
\begin{split}
0&\leqslant C_k^S(\ket{\psi}, \mathcal{C})=\frac{1}{k}\sum_{q|k}\varphi(q)\left[\tr(\rho_S^q) \right]^{\frac{k}{q}} = \frac{1}{k}\sum_{q|k}\varphi(q)\left(\sum_{j=1}^{r}\lambda^1_j \cdot \lambda^{q-1}_j\right)^{\frac{k}{q}}\\
&\leqslant \frac{1}{k}\sum_{q|k}\varphi(q)\lambda_{\mathrm{max}}^{\frac{k(q-1)}{q}} = \frac{1}{k} + \frac{1}{k}\lambda_{\mathrm{max}}^{\frac{k(q-1)}{q}}\sum_{q|k, \ q\geqslant 2}\varphi(q) = \frac{1}{k} + \frac{k-1}{k}\lambda_{\mathrm{max}}^{\frac{k(q-1)}{q}}.
\end{split}
\end{equation}
Since:
\begin{equation}
\lim_{k\rightarrow\infty}\left(\frac{1}{k} + \frac{k-1}{k}\lambda_{\mathrm{max}}^{\frac{k(q-1)}{q}}\right)=0,
\end{equation}
then:
\begin{equation}
\lim_{k\rightarrow\infty}C_k^S(\ket{\psi}, \mathcal{C}) = 0
\end{equation}
Also since:
\begin{equation}
\lim_{k\rightarrow\infty}\left( \left[\tr(\rho_S^2)\right]^{\frac{k-2+(k \bmod 2)}{2}} + \left[\tr(\rho_S^2)\right]^{\frac{k-(k \bmod 2)}{2}} \right) = 0,
\end{equation}
therefore:
\begin{equation}
\lim_{k\rightarrow\infty}C_k^S(\ket{\psi}, \mathcal{D}) = 0.
\end{equation}

\item We begin from the first inequality $1-\tr(\rho_S^k)\geqslant\mathcal{E}_k^S(\ket{\psi}, \mathcal{S})$.
We know that:
\begin{equation}
C_k^S(\ket{\psi}, \mathcal{S}) = \sum_{\substack{g_1+\cdots+g_{r}=k \\ g_i\in\mathbb{Z} \ \mathrm{and} \ 0\leqslant g_i\leqslant k}} \lambda_1^{g_1} \cdots \lambda_r^{g_r}    
\end{equation}
Then the first inequality follows directly by retaining only those terms with $g_j=k$ and $g_i=0$ for all $i\neq j$ and thus:
\begin{equation}
C_k^S(\ket{\psi}, \mathcal{S}) \geqslant \tr(\rho^k),
\end{equation}
and therefore:
\begin{equation}
1-\tr(\rho_S^k)\geqslant\mathcal{E}_k^S(\ket{\psi}, \mathcal{S}).
\end{equation}
For the rest of the inequalities, We start by proving the following lemma:
\begin{lemma}
$P_k^S(\mathcal{C})-P_k^S(\mathcal{D})$ and $P_k^S(\mathcal{D})-P_k^S(\mathcal{S})$ are both projectors.
\label{lemma:diff-proj}
\end{lemma}
\begin{proof}
The Hermicity is obvious. 
Then, note that:
\begin{equation}
\mathcal{D}_k = \{ \mathbb{I},r,r^2,\cdots,r^{k-1}, s,rs,r^2s,\cdots,r^{k-1}s \}     
\end{equation}
and $r^k=\mathbb{I}$, $s^2=\mathbb{I}$, $srs=r^{-1}$ and $r^js = sr^{-j}$.
Then:
\begin{equation}
\begin{split}
\left(P_k^S(\mathcal{C})-P_k^S(\mathcal{D})\right)^2 =& \left(P_k^S(\mathcal{C}) - \frac{1}{2}P_k^S(\mathcal{C}) - \frac{1}{2}P_k^S(\mathcal{C})\cdot s\right)^2 \\
=& \left( \frac{1}{2} P_k^S(\mathcal{C})  - \frac{1}{2}P_k^S(\mathcal{C})\cdot s \right)^2 \\
=&\frac{1}{4}P_k^S(\mathcal{C})  + \frac{1}{4}P_k^S(\mathcal{C})\cdot s \cdot P_k^S(\mathcal{C}) \cdot s- \frac{1}{4} P_k^S(\mathcal{C}) \cdot s - \frac{1}{4} P_k^S(\mathcal{C}) \cdot s \cdot P_k^S(\mathcal{C}) \\
=&\frac{1}{4}P_k^S(\mathcal{C})-\frac{1}{4}P_k^S(\mathcal{C})\cdot s + \frac{1}{4|\mathcal{C}_k|^2} \sum_{j_1,j_2=0}^{k-1}r^{j_1}sr^{j_2}s - \frac{1}{4|\mathcal{C}_k|^2} \sum_{j_1,j_2=0}^{k-1}r^{j_1}sr^{j_2} \\
=& \frac{1}{4}P_k^S(\mathcal{C})-\frac{1}{4}P_k^S(\mathcal{C})\cdot s + \frac{1}{4|\mathcal{C}_k|^2} \sum_{j_1,j_2=0}^{k-1} r^{j_1-j_2} - \frac{1}{4|\mathcal{C}_k|^2}\sum_{j_1,j_2=0}^{k-1} r^{j_1-j_2}s \\
=& \frac{1}{4}P_k^S(\mathcal{C})-\frac{1}{4}P_k^S(\mathcal{C})\cdot s +\frac{1}{4} \left(P_k^S(\mathcal{C})\right)^{2} - \frac{1}{4}\left(P_k^S(\mathcal{C})\right)^{2} \cdot s \\
=& \frac{1}{2} P_k^S(\mathcal{C})  - \frac{1}{2}P_k^S(\mathcal{C})\cdot s \\
=& P_k^S(\mathcal{C})-P_k^S(\mathcal{D}).
\end{split}
\end{equation}
This completes the proof of the idempotence of $P_k^S(\mathcal{C})-P_k^S(\mathcal{D})$, thus a projector. 
Next, for $P_k^S(\mathcal{D})-P_k^S(\mathcal{S})$, since for any permutation $\pi$, $\pi \cdot P_k^S(\mathcal{S}) = P_k^S(\mathcal{S}) \cdot \pi = P_k^S(\mathcal{S})$, therefore:
\begin{equation}
\begin{split}
\left(P_k^S(\mathcal{D})-P_k^S(\mathcal{S})\right)^2 
=& P_k^S(\mathcal{D}) + P_k^S(\mathcal{S}) - P_k^S(\mathcal{D})P_k^S(\mathcal{S}) - P_k^S(\mathcal{S})P_k^S(\mathcal{D}) \\
=& P_k^S(\mathcal{D}) + P_k^S(\mathcal{S}) - P_k^S(\mathcal{S}) - P_k^S(\mathcal{S}) \\
=& P_k^S(\mathcal{D})-P_k^S(\mathcal{S}).
\end{split}
\end{equation}
\end{proof}
Since the projectors have only eigenvalues 0 and 1, then $P_k^S(\mathcal{C})-P_k^S(\mathcal{D})$ and $P_k^S(\mathcal{D})-P_k^S(\mathcal{S})$ are both positive semidefinite Hermitian matrices.
Then:
\begin{equation}
\begin{split}
\mathcal{E}_k^S(\ket{\psi}, \mathcal{S}) -\mathcal{E}_k^S(\ket{\psi}, \mathcal{D}) = C_k^S(\ket{\psi}, \mathcal{D}) - C_k^S(\ket{\psi}, \mathcal{S}) = \tr\left(  \left( P_k^S(\mathcal{D}) \otimes \mathbb{I}_k^{S^c} - P_k^S(\mathcal{S}) \otimes \mathbb{I}_k^{S^c} \right) \rho^{\otimes k}  \right)\geqslant 0,
\end{split}
\end{equation}
and:
\begin{equation}
\begin{split}
\mathcal{E}_k^S(\ket{\psi}, \mathcal{D}) -\mathcal{E}_k^S(\ket{\psi}, \mathcal{C}) = C_k^S(\ket{\psi}, \mathcal{C}) - C_k^S(\ket{\psi}, \mathcal{D}) = \tr\left(  \left( P_k^S(\mathcal{C}) \otimes \mathbb{I}_k^{S^c} - P_k^S(\mathcal{D}) \otimes \mathbb{I}_k^{S^c} \right) \rho^{\otimes k}  \right)\geqslant 0.
\end{split}
\end{equation}
Then finally:
\begin{equation}
1-\tr(\rho_S^k)\geqslant\mathcal{E}_k^S(\ket{\psi}, \mathcal{S})\geqslant\mathcal{E}_k^S(\ket{\psi}, \mathcal{D})\geqslant\mathcal{E}_k^S(\ket{\psi}, \mathcal{C}).
\end{equation}
\end{enumerate}

\section{Sampling noise analysis\label{sm:noise}}

In this section we provide a detailed sampling noise analysis under the condition that the number of state copies is limited.
Apart from the absolute error with fixed $k$ and symmetric group in Fig.~\ref{fig:error}, we also provide numerics of all three groups for both absolute and logarithmic errors with either fixed $k$ or fixed $N_{\mathrm{tot}}$, as shown in Fig.~\ref{fig:log-error} and Fig.~\ref{fig:fixed-Ntot}.
Note that the error scaling with respect to the number of state copies we derived may not be tight in terms of order $k$.
However, for fixed $k$, both absolute error $\varepsilon$ and relative error $\eta$ show the error behavior of $O(N^{-1/2}_{\mathrm{tot}})$, which matches our analytics.

\begin{figure*}
    \centering
    \includegraphics[width=0.9\linewidth]{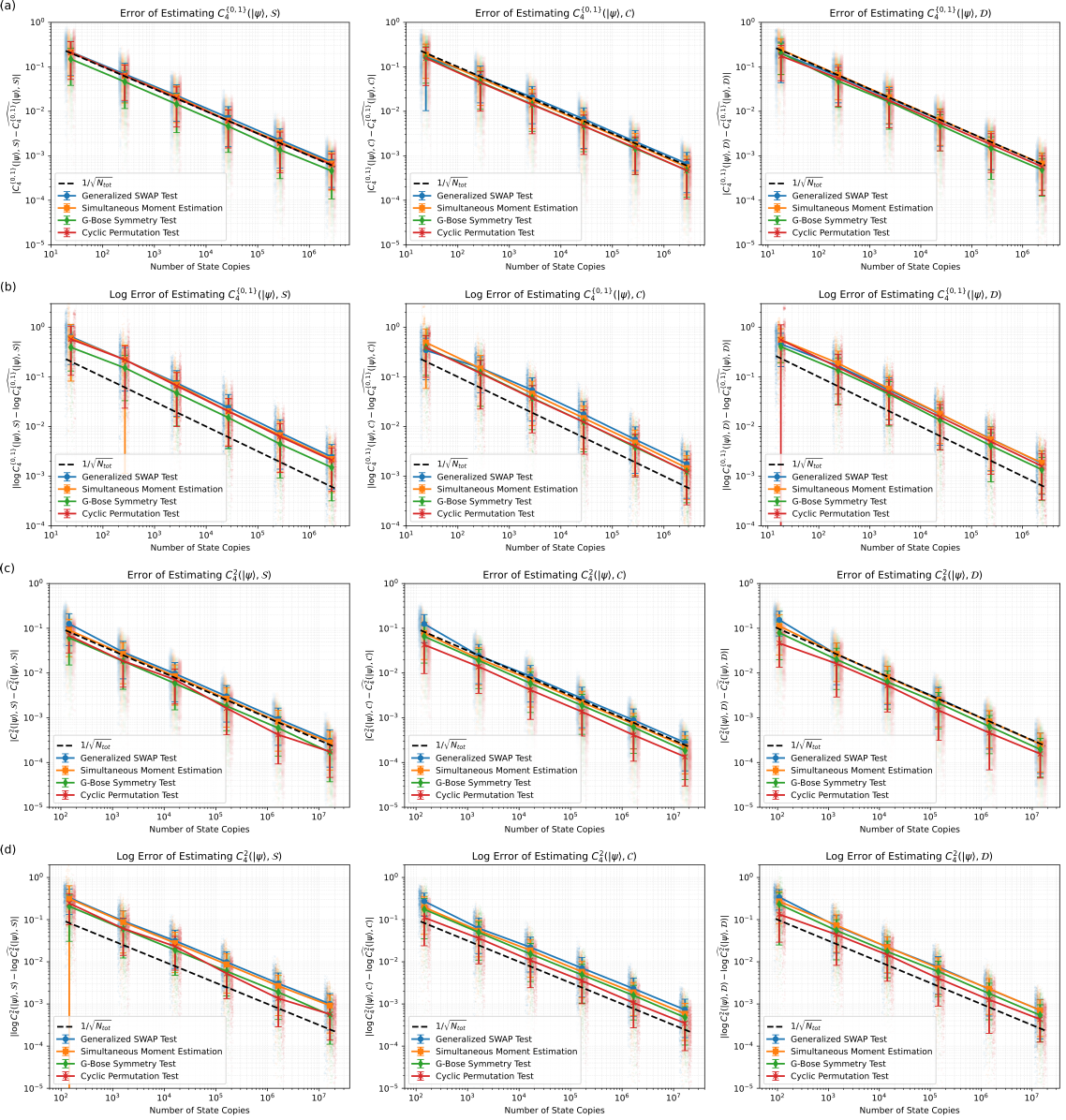}
    \caption{Absolute error and logarithmic error in estimating $C_4^{\{0,1\}}(\ket{\psi})$ and $C_4^{2}(\ket{\psi})$. 
    $S=\{0,1\}$ represents the subsystem consisting of the first two qubits.
    The numerical settings are the same as the ones in Fig.~\ref{fig:error}.
    The empirical absolute errors have the scaling very close to $\sim N_{\mathrm{tot}}^{-1/2}$ for all three groups and both (a) bipartite and (c) multipartite cases.
    The empirical logarithmic errors also exhibit the scaling $\varepsilon \sim O( N_{\mathrm{tot}}^{-1/2})$, though with a factor larger than 1, as shown in (b,d).}
    \label{fig:log-error}
\end{figure*}

\begin{figure*}
    \centering
    \includegraphics[width=0.9\linewidth]{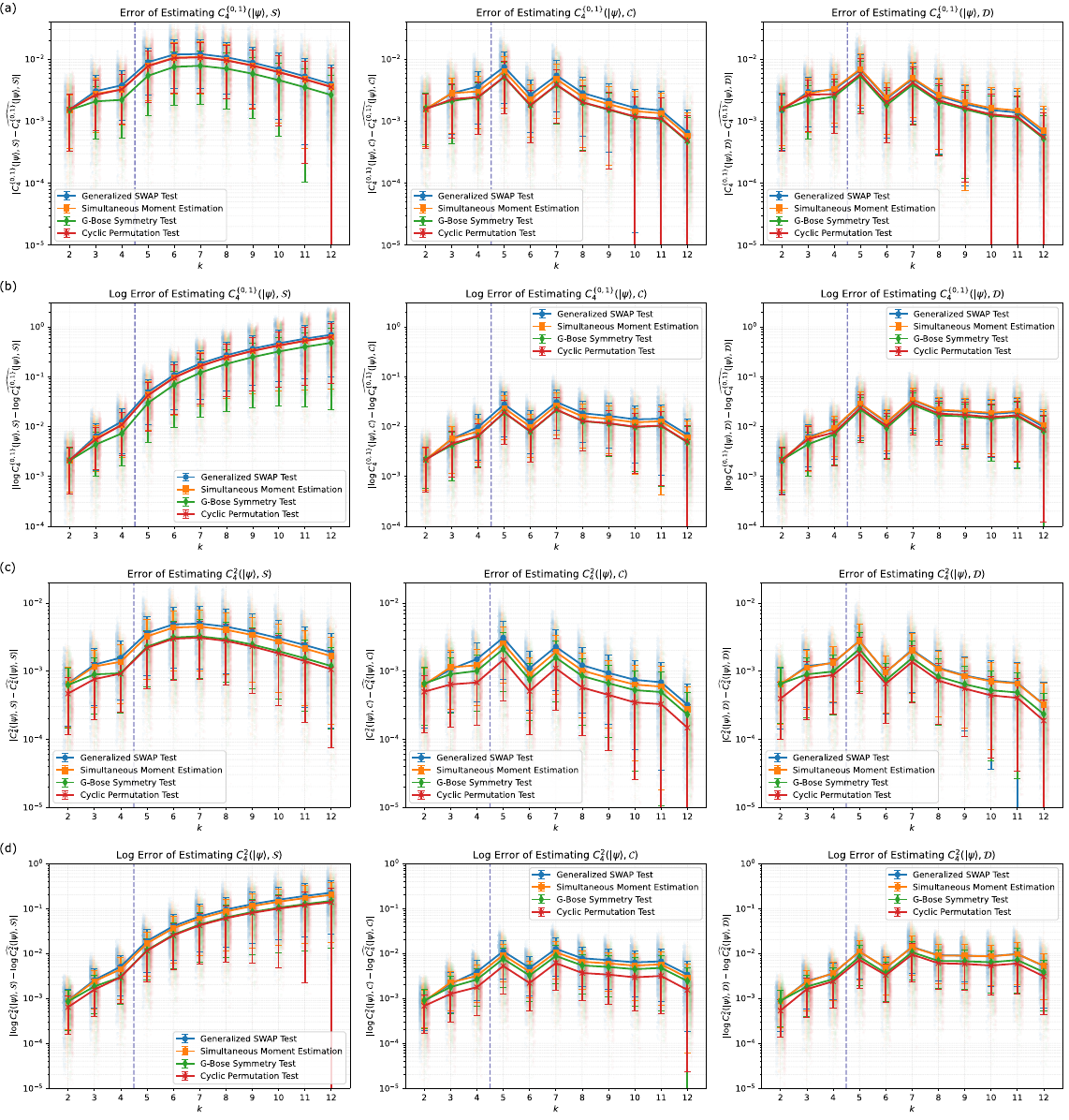}
    \caption{Absolute and logarithmic error in estimating (a,b) $C_4^{\{0,1\}}(\ket{\psi})$ and (c,d) $C_4^{2}(\ket{\psi})$ with respect to $k$, for $N_{\mathrm{tot}}=100000$ and $600000$, respectively. 
    Other numerical settings are the same as the ones in Fig.~\ref{fig:error}.
    We use the Newton-Girard method to extrapolate higher-order state moments from the estimates at $k=2,3,4$, thereby obtaining $C_k^{\{0,1\}}(\ket{\psi})$ or $C_k^{2}(\ket{\psi})$ for $k\geq 5$ without consuming additional state copies (right of the vertical dashed line).}
    \label{fig:fixed-Ntot}
\end{figure*}

\subsection{Generalized SWAP test}

We start from the generalized SWAP test.
The circuit is a special case of Hadamard test, where the control unitary is:
\begin{equation}
    \ket{0}\bra{0} \otimes \mathbb{I} + \ket{1}\bra{1}\otimes (1 \ 2 \ \cdots \ k),
\end{equation}
applied across all $k$ copies of subsystems $S$.
By sampling the outcome of the single auxiliary qubit, we naturally have:
\begin{equation}
    \tr(\rho^k) = 2p\left( \ket{0} \right) - 1.
\end{equation}
We denote $M_i=1$ if the auxiliary qubit is measured at $\ket{0}$, and vice versa for $M_i=0$.
Then the estimated state moment via $N$ times of measurements is:
\begin{equation}
    \widehat{\tr(\rho^k)} = 2 \frac{1}{N}\sum_{i=1}^{N}M_i - 1.
\end{equation}
Since from Hoeffding's inequality:
\begin{equation}
    p\left( \left| \frac{1}{N}\sum_{i=1}^{N}M_i - \frac{\tr(\rho^k)+1}{2} \right|\geqslant \varepsilon \right) \leqslant 2\exp (-2N\varepsilon^2),
\end{equation}
then:
\begin{equation}
    p\left( \left| \widehat{\tr(\rho^k)} -\tr(\rho^k)  \right| \geqslant \varepsilon \right) \leqslant 2\exp\left( -\frac{N\varepsilon^2}{2} \right).
\end{equation}
Then, we state a lemma:
\begin{lemma}
For a multi-variable continuous function $g: \ \mathbb{R}^m\rightarrow\mathbb{R}$, we have:
\begin{equation}
    |g(\mathbf{y})-g(\mathbf{x})| \leqslant \sum_{i=1}^{m} \sup_{\xi\in[\mathbf{x},\mathbf{y}]}\left| \frac{\partial g}{\partial z_i}(\xi) \right| |y_i-x_i|,\nonumber
\end{equation}
where $\xi$ denotes any point on the line segment between $\mathbf{x}$ and $\mathbf{y}$. $\frac{\partial g}{\partial z_i}(\xi)$ means by taking the partial derivative on the $i$-th axis and evaluating it at the point $\xi$.
\label{lemma:mean-value}
\end{lemma}
\begin{proof}
We define a function $h: \ \mathbb{R}\rightarrow\mathbb{R}$ on $0\leqslant t\leqslant1$ such that:
\begin{equation}
h(t) = g(\mathbf{x}+t(\mathbf{y} - \mathbf{x})).
\end{equation}
Then:
\begin{equation}
h'(t) = \nabla g(\mathbf{x}+t(\mathbf{y} - \mathbf{x})) \cdot (\mathbf{y} - \mathbf{x}).
\end{equation}
By integration, we have:
\begin{equation}
g(\mathbf{y}) - g(\mathbf{x}) = \int_{0}^{1}\nabla g(\mathbf{x}+t(\mathbf{y} - \mathbf{x})) \cdot (\mathbf{y} - \mathbf{x}) dt.
\end{equation}
Therefore:
\begin{equation}
\begin{split}
&|g(\mathbf{y}) - g(\mathbf{x})| = \left| \int_{0}^{1}\nabla g(\mathbf{x}+t(\mathbf{y} - \mathbf{x})) \cdot (\mathbf{y} - \mathbf{x}) dt \right| \leqslant \int_{0}^{1} \left| \nabla g(\mathbf{x}+t(\mathbf{y} - \mathbf{x})) \cdot (\mathbf{y} - \mathbf{x}) \right|dt \\
\leqslant& \sum_{i=1}^{m} |y_i - x_i| \int_{0}^{1} \left|\frac{\partial g}{\partial z_i}(\mathbf{x}+t(\mathbf{y}-\mathbf{x}))\right| dt\leqslant \sum_{i=1}^{m} \sup_{\xi\in[\mathbf{x},\mathbf{y}]}\left| \frac{\partial g}{\partial z_i}(\xi) \right| |y_i-x_i|.
\end{split}
\end{equation}
Note that this can also be viewed as the mean value theorem (inequality) generalized to the multivariate case $\mathbb{R}^m\rightarrow\mathbb{R}$~\cite{hormander2003analysis}.
\end{proof}
We consider $C_k^S(\ket{\psi})$ is a multivariate function with $\{\tau_l=\tr(\rho^l_S)\}_{l=2}^{k}$ (we exclude $l=1$ since $\tau_1=1$).
Then:
\begin{equation}
\left|\frac{\partial C_k^S}{\partial \tau_l}\right| = \left|\frac{1}{|\mathcal{G}_k|}\sum_{\vec{\mathbf{k}} \vdash k} N_{\mathcal{G}}(\vec{\mathbf{k}}) m_l\tau_l^{m_l-1} \prod_{j\neq l}\tau_j^{m_j}\right|\leqslant\frac{1}{|\mathcal{G}_k|}\sum_{\vec{\mathbf{k}} \vdash k} N_{\mathcal{G}}m_l,
\end{equation}
and suppose $|\tau_l - \widehat{\tau_l}|\leqslant \varepsilon_l$ with shots $N_l$, then:
\begin{equation}
\left|C_k^S(\tau_2,\cdots,\tau_k) - C_k^S(\widehat{\tau_2},\cdots,\widehat{\tau_k})\right| \leqslant \sum_{l=2}^{k} \sup_{\xi\in[\mathbf{x},\mathbf{y}]}\left| \frac{\partial C_k^S}{\partial \tau_l}(\xi) \right| |\tau_l - \widehat{\tau_l}| \leqslant \sum_{l=2}^{k}\alpha_{l} \varepsilon_l,
\end{equation}
where:
\begin{equation}
\alpha_{l} \geqslant \sup_{\xi\in[\mathbf{x},\mathbf{y}]}\left| \frac{\partial C_k^S}{\partial \tau_l}(\xi) \right|.
\end{equation}
Using the union bound, we request that the probability of having relatively large errors is smaller than a finite probability $\delta$:
\begin{equation}
    p\left( \bigcup_{l=2}^{k}\left|\widehat{\tau_l} - \tau_l\right|\geqslant \varepsilon_l  \right) \leqslant \sum_{l=2}^{k}p\left( \left|\widehat{\tau_l} - \tau_l\right|\geqslant \varepsilon_l  \right) \leqslant 2\sum_{l=2}^{k}\exp\left( -\frac{N_l\varepsilon_l^2}{2} \right)\leqslant \delta.
\end{equation}
In this case, if $\sum_{l=2}^{k}\alpha_l\varepsilon_l\leqslant \varepsilon$ we can say the absolute error of estimating $C_k^S(\ket{\psi})$ can be achieved below $\varepsilon$ with confidence $1-\delta$.
A sufficient condition to achieve $(1-\delta)$-confidence is the following,
\begin{equation}
    2\exp\left( -\frac{N_l\varepsilon_l^2}{2} \right) \leqslant \frac{\delta}{k-1},
\end{equation}
where for simplicity we split the $\delta$ equally.
One can also take variational $\delta_l$ such that $\sum_{l}\delta_l$ and get similar bounds eventually.
Then:
\begin{equation}
    N_l \geqslant \frac{2}{\varepsilon_l^2} \log \frac{2(k-1)}{\delta}.
\end{equation}
Then the required total number of state copies is:
\begin{equation}
    N_{\mathrm{tot}}=\sum_{l=2}^{k}lN_l \geqslant 2\log\frac{2(k-1)}{\delta}\sum_{l=2}^{k}\frac{l}{\varepsilon_l^2}.
\end{equation}
We now minimize $\sum_{l=2}^{k}l /\varepsilon_l^2$ with the inequality condition $\sum_{l=2}^{k}\alpha_l\varepsilon_l\leqslant \varepsilon$.
The minimized result will yield the optimal lower bound of the number of state copies, as well as the distribution of each $N_l$.
For simplicity we set $\sum_{l=2}^k\alpha_l\varepsilon_l = \varepsilon$.
To make this bound as tight as possible, we will optimize the right hand side with the Lagrange multiplier method, where we define:
\begin{equation}
\mathcal{L}(\varepsilon_1,\cdots,\varepsilon_l, \lambda) = \sum_{l=2}^k\frac{l}{\varepsilon_l^2} + \lambda \left( \sum_{l=2}^{k}\alpha_l\varepsilon_l-\varepsilon \right).
\end{equation}
Then, for the minimum:
\begin{equation}
\frac{\partial\mathcal{L}(\varepsilon_1,\cdots,\varepsilon_l, \lambda)}{\partial \varepsilon_l} = -\frac{2l}{\varepsilon_l^3}+\lambda\alpha_l=0.
\end{equation}
Therefore:
\begin{equation}
\varepsilon_l = \left(\frac{2l}{\lambda \alpha_l}\right)^{1/3}.
\end{equation}
Since:
\begin{equation}
\varepsilon = \sum_{l=2}^{k}\alpha_l\varepsilon_l = \sum_{l=2}^{k} \alpha_l\left(\frac{2l}{\lambda \alpha_l}\right)^{1/3} = \left(\frac{2}{\lambda}\right)^{1/3} \sum_{l=2}^{k}\alpha_l^{2/3} l^{1/3},
\end{equation}
then:
\begin{equation}
\lambda = \frac{2\left( \sum_{l=2}^k \alpha_l^{2/3} l^{1/3} \right)^3}{\varepsilon^3}.
\end{equation}
Then:
\begin{equation}
N_{\mathrm{tot}}=\sum_{l=2}^{k}lN_l\geqslant 2\log\frac{2(k-1)}{\delta} \sum_{l=2}^{k} l \left( \frac{2l}{\lambda \alpha_l} \right)^{-2/3} = \frac{2}{\varepsilon^2} \log \frac{2(k-1)}{\delta} \left( \sum_{l=2}^{k}l^{1/3}\alpha_l^{2/3} \right)^3,
\end{equation}
with:
\begin{equation}
N_l \geqslant \frac{2}{\varepsilon_l^2}\log\frac{2(k-1)}{\delta} = \frac{2}{\varepsilon^2}\log\frac{2(k-1)}{\delta} \left( \frac{\alpha_l}{l} \right)^{2/3} \left( \sum_{q=2}^{k}\alpha_q^{2/3} q^{1/3} \right)^2.
\end{equation}
Note that for multipartite cases, it is easy to see that the scaling remains the same but maybe with the extra coefficient.

For logarithmic error such that $\left|\log\left( \widehat{C_k^S(\ket{\psi})} \right) - \log \left( C_k^S(\ket{\psi}) \right)\right| \leqslant \eta$, from Chernoff bound we have:
\begin{equation}
\begin{split}
p\left( \sum_{i=1}^N\frac{M_i}{N} \leqslant (1-w) \mathbb{E}\left[ \sum_{i=1}^N\frac{M_i}{N} \right] \right) &\leqslant \exp \left( -\frac{w^2N\mathbb{E}\left[ \sum_{i=1}^N\frac{M_i}{N} \right]}{2} \right),\\
p\left( \sum_{i=1}^N\frac{M_i}{N} \geqslant (1+w) \mathbb{E}\left[ \sum_{i=1}^N\frac{M_i}{N} \right] \right) &\leqslant \exp \left( -\frac{w^2N\mathbb{E}\left[ \sum_{i=1}^N\frac{M_i}{N} \right]}{2+w} \right).
\end{split}
\end{equation}
Assuming we desire $\left|\widehat{\tr(\rho^k)} - \tr(\rho^k)\right| \leqslant \varepsilon$, then:
\begin{equation}
\left| \frac{1}{N}\sum_{i=1}^{N}M_i - \frac{\tr(\rho^k)+1}{2} \right|\leqslant \frac{\varepsilon}{2}.
\end{equation}
We set $(1\pm w)\mathbb{E}\left[ \sum_{i=1}^N\frac{M_i}{N} \right] = \mathbb{E}\left[ \sum_{i=1}^N\frac{M_i}{N} \right] \pm \frac{\varepsilon}{2}$, then:
\begin{equation}
w = \frac{\varepsilon}{2 \mathbb{E}\left[ \sum_{i=1}^N\frac{M_i}{N} \right]}.
\end{equation}
Therefore:
\begin{equation}
\begin{split}
p\left( \left| \widehat{\tau_k} -\tau_k  \right| \geqslant \varepsilon \right)  =&p\left( \sum_{i=1}^N\frac{M_i}{N} \leqslant (1-w) \mathbb{E}\left[ \sum_{i=1}^N\frac{M_i}{N} \right] \right) + p\left( \sum_{i=1}^N\frac{M_i}{N} \geqslant (1+w) \mathbb{E}\left[ \sum_{i=1}^N\frac{M_i}{N} \right] \right) \\
\leqslant& \exp \left( -\frac{w^2N\mathbb{E}\left[ \sum_{i=1}^N\frac{M_i}{N} \right]}{2} \right) + \exp \left( -\frac{w^2N\mathbb{E}\left[ \sum_{i=1}^N\frac{M_i}{N} \right]}{2+w} \right) \\
=& \exp\left( -\frac{\varepsilon^2 N}{8\mathbb{E}\left[ \sum_{i=1}^N\frac{M_i}{N} \right]} \right) + \exp\left( -\frac{\varepsilon^2 N}{8\mathbb{E}\left[ \sum_{i=1}^N\frac{M_i}{N} \right]+2\varepsilon} \right) \\
=& \exp\left( -\frac{\varepsilon^2 N}{4(1+\tau_k)} \right) + \exp\left( -\frac{\varepsilon^2 N}{4(1+\tau_k)+2\varepsilon} \right).
\end{split}
\end{equation}
Then we request that the probability of error is smaller or equal than a finite failure probability $\vartheta$:
\begin{equation}
\begin{split}
p\left( \bigcup_{l=2}^{k}\left|\widehat{\tau_l} - \tau_l\right|\geqslant \varepsilon_l  \right) 
\leqslant& \sum_{l=2}^{k}p\left( \left|\widehat{\tau_l} - \tau_l\right|\geqslant \varepsilon_l  \right) \\
\leqslant& \sum_{l=2}^{k}\left(\exp\left( -\frac{\varepsilon_l^2 N_l}{4(1+\tau_l)} \right) + \exp\left( -\frac{\varepsilon_l^2 N_l}{4(1+\tau_l)+2\varepsilon_l} \right)\right) \\
\leqslant& 2\sum_{l=2}^{k}\exp\left( -\frac{\varepsilon_l^2 N_l}{4(1+\tau_l)+2} \right) \\
\leqslant& \vartheta.
\end{split}
\end{equation}
For simplicity we again split the failure probability $\vartheta$ equally, i.e.,:
\begin{equation}
2\exp\left( -\frac{\varepsilon_l^2 N_l}{4(1+\tau_l)+2} \right) \leqslant \frac{\vartheta}{k-1}.
\end{equation}
Therefore we obtain a lower bound on the number of shots per experiment,
\begin{equation}
N_l\geqslant \frac{6+4\tau_l}{\varepsilon_l^2}\log\frac{2(k-1)}{\vartheta}.
\end{equation}
Since:
\begin{equation}
\left|\log\left( \widehat{C_k^S(\ket{\psi})} \right) - \log \left( C_k^S(\ket{\psi}) \right)\right| \leqslant \eta,
\end{equation}
then:
\begin{equation}
\left|\widehat{C_k^S(\ket{\psi})} -  C_k^S(\ket{\psi}) \right| \leqslant (1-e^{-\eta}) \max\{\widehat{C_k^S(\ket{\psi})}, C_k^S(\ket{\psi})\}.
\end{equation}
We set $\sum_{l=2}^{k}\alpha_l\varepsilon_l = (1-e^{-\eta}) \max\{\widehat{C_k^S(\ket{\psi})}, C_k^S(\ket{\psi})\}$ and we optimize:
\begin{equation}
N_{\mathrm{tot}}=\sum_{l=2}^{k}lN_l \geqslant 2\log\frac{2(k-1)}{\vartheta}\sum_{l=2}^{k}l\frac{3+2\tau_l}{\varepsilon_l^2}.
\end{equation}
Similarly, we use Lagrange multiplier approach and we finally have:
\begin{equation}
N_{\mathrm{tot}}=\sum_{l=2}^{k} lN_l \geqslant 2\log\frac{2(k-1)}{\vartheta}\frac{\left(\sum_{l=2}^{k}(3+2\tau_l)^{1/3}l^{1/3}\alpha_l^{2/3}\right)^3}{\left((1-e^{-\eta}) \max\{\widehat{C_k^S(\ket{\psi})}, C_k^S(\ket{\psi})\}\right)^2},
\end{equation}
with:
\begin{equation}
\varepsilon_l = \left(\frac{2l(3+2\tau_l)}{\lambda \alpha_l}\right)^{1/3},
\end{equation}
and:
\begin{equation}
\lambda = 2 \left( \frac{\sum_{l=2}^{k}(3+2\tau_l)^{1/3}l^{1/3}\alpha_l^{2/3}}{(1-e^{-\eta}) \max\{\widehat{C_k^S(\ket{\psi})}, C_k^S(\ket{\psi})\}} \right)^3.
\end{equation}
Note that this bound depends on $\tau_l$, which we do not know at the beginning. 
One can make it state-independent by taking the worst scenario where $\tau_l=1$.
Moreover, one can also make trials with small samples to roughly estimate $\tau_l$ before making large number of circuit executions.
For simplicity we choose the former case, therefore:
\begin{equation}
N_{\mathrm{tot}}=\sum_{l=2}^{k}lN_l \geqslant \frac{10}{\left((1-e^{-\eta}) \max\{\widehat{C_k^S(\ket{\psi})}, C_k^S(\ket{\psi})\}\right)^2}\log\frac{2(k-1)}{\vartheta}\left( \sum_{l=2}^{k}\alpha_l^{2/3}l^{1/3} \right)^3.
\end{equation}
Note that when $\eta$ is relatively small, $(1-e^{-\eta}) \sim\eta$, showing that the logarithmic error $\eta$ again follows $O(N_{\mathrm{tot}}^{-1/2})$, as shown in Fig.~\ref{fig:log-error}(b).
Also, similarly, the same scaling can also be found in multipartite cases but maybe also with extra coefficient.

We now list values of $\alpha_l$ for $\mathcal{G}=\mathcal{S},\mathcal{C},\mathcal{D}$:
\begin{enumerate}
\item For $\mathcal{G}=\mathcal{S}$, we have:
\begin{equation}
\left| \frac{\partial C_k^S(\ket{\psi}, \mathcal{S})}{\partial \tau_l} \right| = \frac{1}{k!} \sum_{\pi \in \mathcal{S}_k} m_l(\pi) \tau_l^{m_l(\pi)-1} \prod_{j\neq l}\tau_j^{m_j(\pi)} \leqslant \frac{1}{k!} \sum_{\pi\in\mathcal{S}_k}m_l(\pi) = \mathbb{E}\left[ m_l(\pi) \right] = \frac{1}{l}.
\end{equation}
Then we can set:
\begin{equation}
\alpha_l(\mathcal{S}) = \frac{1}{l}.  
\label{eq:alpha-S}
\end{equation}
Note that this is also true when applying the recurrence relation in Eq.~\eqref{eq:recur} since from:
\begin{equation}
C_k^S(\ket{\psi}, \mathcal{S}) = \frac{1}{k}\sum_{q=0}^{k-1}C_q^S(\ket{\psi}) \tau_{k-q}, 
\end{equation}
we again take the partial derivatives of $C_k^S(\ket{\psi}, \mathcal{S})$ with respect to $\{\tau_l\}_{l=2}^{k}$.
Then:
\begin{equation}
\left| \frac{\partial C_k^S(\ket{\psi}, \mathcal{S})}{\partial \tau_l} \right| = \frac{1}{k} \left| C_{k-l}^S(\ket{\psi}) + \sum_{q=l}^{k-1}\frac{\partial C_q^S(\ket{\psi}, \mathcal{S})}{\partial \tau_l} \tau_{k-q} \right|.
\end{equation}
We then use induction method to prove that $\alpha_l(\mathcal{S})=l^{-1}$ also holds when we apply Eq.~\eqref{eq:recur}.
This statement of course holds for $2\leqslant k\leqslant l-1$ since the derivative yields 0.
Then, for $k=l$:
\begin{equation}
0\leqslant\frac{\partial C_k^S(\ket{\psi}, \mathcal{S})}{\partial \tau_l} = \frac{1}{l}C_0^S(\ket{\psi})=\frac{1}{l}
\end{equation}
Then for $k= l+1$: 
\begin{equation}
\frac{\partial C_k^S(\ket{\psi}, \mathcal{S})}{\partial \tau_l} = \frac{1}{l+1}  \left(C_{k+1-l}^S(\ket{\psi}) + \frac{\partial C_q^S(\ket{\psi}, \mathcal{S})}{\partial \tau_l} \tau_{k+1-q} \right) \leqslant \frac{1}{l+1} \left( 1 + \frac{1}{l} \right) = \frac{1}{l}.
\end{equation}
Therefore, again:
\begin{equation}
\left| \frac{\partial C_k^S(\ket{\psi}, \mathcal{S})}{\partial \tau_l} \right| \leqslant \frac{1}{l}.
\end{equation}
Then we can again take $\alpha_l(\mathcal{S})=1/l$.
\item For $\mathcal{G}=\mathcal{C}$, when $l|k$ we have:
\begin{equation}
\left| \frac{\partial C_k^S(\ket{\psi}, \mathcal{C})}{\partial \tau_l} \right| = \left| \frac{1}{k}\varphi(l)\frac{k}{l}\tau^{k/l-1}_l \right| = \frac{\varphi(l)}{l}\tau^{k/l-1}_l \leqslant\frac{\varphi(l)}{l}.
\end{equation}
Therefore we take $\alpha_l = \varphi(l) / l$ if $l|k$, otherwise $\alpha_l=0$, i.e.,
\begin{equation}
\alpha_l(\mathcal{C}) = \frac{\varphi(l)}{l}\boldsymbol{\delta}(l|k).
\label{eq:alpha-C}
\end{equation}
\item Finally, for $\mathcal{G}=\mathcal{D}$, if $k$ is even:
\begin{equation}
\frac{\partial}{\partial \tau_2}\left( \frac{1}{4}\tau_2^{k/2-1}+\frac{1}{4}\tau_2^{k/2} \right) = \frac{1}{4}\left( \frac{k}{2}-1 \right)\tau_2^{k/2-2} + \frac{1}{4}\frac{k}{2}\tau_2^{k/2-1}\leqslant \frac{k-1}{4}.
\end{equation}
If $k$ is odd:
\begin{equation}
\frac{\partial}{\partial \tau_2} \left( \frac{1}{2}\tau_2^{(k-1)/2} \right) = \frac{k-1}{4}\tau_2^{(k-3)/2} \leqslant\frac{k-1}{4}.
\end{equation}
Therefore, we take:
\begin{equation}
\alpha_l(\mathcal{D}) = \frac{\varphi(l)}{2l}\boldsymbol{\delta}(l|k) + \frac{k-1}{4}\boldsymbol{\delta}(l=2).
\label{eq:alpha-D}
\end{equation}
\end{enumerate}
Following the above assigned $\alpha_l$ in Eq.~\eqref{eq:alpha-S},~\eqref{eq:alpha-C} and~\eqref{eq:alpha-D}, one can easily obtain the near-optimal allocation strategies shown in Table~\ref{tab:alloc}.

\subsection{Simultaneous Moment Estimation}

Recently, a near-optimal protocol was proposed to simultaneously estimate the moments $\tr(\rho_S^{l})$ for fixed $S$ and $l$ from 2 to $k$~\cite{shi2025nearoptimal}.
The scheme outputs all moments with additive error at most $\varepsilon_0$ and success probability at least $2/3$, using $O\left(\frac{k}{\varepsilon_0^2}\log k\right)$ copies of $\rho_S$.
The circuit is shown in Fig.~\ref{fig:circuits}(b), and the post-processing that maps measurement outcomes to each moment is given in Algorithm~1 of~\cite{shi2025nearoptimal}.
Consequently, fixing the success probability at $2/3$, the total number of state copies required is:
\begin{equation}
N_{\mathrm{tot}} \sim O\left( \left( \sum_{l=2}^{k}\alpha_l \right)^2 \frac{k}{\varepsilon^2}\log k \right) \ \ \text{or} \ \ O\left( \left( \sum_{l=2}^{k}\alpha_l \right)^2 \frac{k}{\left(\eta \max\{\widehat{C_k^S(\ket{\psi})}, C_k^S(\ket{\psi})\}\right)^2}\log k \right),
\end{equation}
thus both $\varepsilon$ and $\eta$ also following the standard scaling $O(N^{-1/2}_{\mathrm{tot}})$.

\subsection{G-Bose Symmetry Test}

Then we consider the scenarios of G-Bose symmetry test.
Starting from bipartite case, notably:
\begin{equation}
    p(\ket{\mathbf{0}}) = C_k^S(\ket{\psi}).
\end{equation}
Every time we execute the circuit, we have the measurement outcome either $\ket{\mathbf{0}}$ or not. 
Consider $N$ times of circuit executions and for each execution we denote $M=1$ if the outcome is $\ket{\mathbf{0}}$, and $M=0$ for the if the outcome is not $\ket{\mathbf{0}}$.
Then $M$ are independent random binary variables following the probability $\{p(\ket{\mathbf{0}}), 1-p(\ket{\mathbf{0}})\}$.
Then:
\begin{equation}
    \widehat{p(\ket{\mathbf{0}})} = \frac{M_1+M_2+\cdots+M_N}{N}
\end{equation}
Thus, the absolute error becomes:
\begin{equation}
    \left| \widehat{C_k^S(\ket{\psi})} - C_k^S(\ket{\psi}) \right| = \left| \widehat{p(\ket{\mathbf{0}})}- C_k^S(\ket{\psi})  \right| = \left| \frac{M_1+M_2+\cdots+M_N}{N}- C_k^S(\ket{\psi})  \right|
\end{equation}
From Hoeffding's inequality, we have:
\begin{equation}
    p\left( \left| \widehat{C_k^S(\ket{\psi})} - C_k^S(\ket{\psi}) \right| \geqslant \varepsilon \right)\leqslant 2\exp \left(  -\frac{2\varepsilon^2}{\frac{1}{N^2}N} \right)=2\exp\left(-2N\varepsilon^2\right)\leqslant\delta.
\end{equation}
Therefore, in order to get absolute error smaller than $\varepsilon$ with confidence $1-\delta$, one needs:
\begin{equation}
    N\geqslant\frac{\log\frac{2}{\delta}}{2\varepsilon^2}, 
\end{equation}
which consumes at least in total:
\begin{equation}
    N_{\mathrm{tot}}=kN \geqslant k\frac{\log\frac{2}{\delta}}{2\varepsilon^2},
\end{equation}
number of copies of $\ket{\psi}$.
Then for multipartite cases, we denote each $S$ such that $|S|=s$ as $\{S_i\}_{i=1}^{\binom{n}{s}}$. 
For simplicity, we assume that we allocate $N_0$ times of executions per $C_k^S(\ket{\psi})$, then:
\begin{equation}
\begin{split}
    \left| \widehat{C_k^s(\ket{\psi})} - C_k^s(\ket{\psi}) \right| = \left| \sum_{i=1}^{\binom{n}{s}}\frac{M^{(S_i)}_1+\cdots+M^{(S_{i})}_{N_0}}{\binom{n}{s}N_0} - C_k^s(\ket{\psi}) \right|,
\end{split}
\end{equation}
where $M^{(S_i)}_j$ denotes the $M$ for $j$-th executed circuit with bipartition $S_i|S_i^c$.
Then, $M^{(S_i)}_j\leqslant N_0^{-1}\binom{n}{s}^{-1}$.
Again, from Hoeffding's inequality:
\begin{equation}
    p\left( \left| \widehat{C_k^s(\ket{\psi})} - C_k^s(\ket{\psi}) \right| \geqslant \varepsilon \right)\leqslant 2\exp \left(  -\frac{2\varepsilon^2}{ N_0^{-1}\binom{n}{s}^{-1}} \right)=2\exp\left(-2N_0\binom{n}{s}\varepsilon^2\right)\leqslant\delta.
\end{equation}
Therefore, in order to get absolute error smaller than $\varepsilon$ with confidence $1-\delta$, one needs:
\begin{equation}
    N=N_0\binom{n}{s}\geqslant \frac{\log \frac{2}{\delta}}{2\varepsilon^2}.
\end{equation}
Therefore, this requires the same lower-bounded number of state copies as the bipartite case:
\begin{equation}
    N_{\mathrm{tot}}=kN\geqslant k\frac{\log\frac{2}{\delta}}{2\varepsilon^2}.
\end{equation}
Then, for relative errors, from Chernoff bound, we have:
\begin{equation}
\begin{split}
    &p\left( \sum_{i=1}^{N}M_i\geqslant(1+w)\mathbb{E}\left[ \sum_{i=1}^{N}M_i \right]   \right) \leqslant \exp \left( -\frac{w^2\mathbb{E}\left[ \sum_{i=1}^{N}M_i \right]}{2+w}  \right), \\
    &p\left( \sum_{i=1}^{N}M_i\leqslant(1-w)\mathbb{E}\left[ \sum_{i=1}^{N}M_i \right]   \right) \leqslant \exp \left( -\frac{w^2\mathbb{E}\left[ \sum_{i=1}^{N}M_i \right]}{2}  \right).
\end{split}
\end{equation}
Therefore:
\begin{equation}
\begin{split}
    &p\left( \sum_{i=1}^{N}\frac{M_i}{N}\geqslant(1+w)C_k^S(\ket{\psi})   \right) \leqslant \exp \left( -\frac{w^2 N C_k^S(\ket{\psi})}{2+w}  \right)\leqslant \vartheta, \\
    &p\left( \sum_{i=1}^{N}\frac{M_i}{N}\leqslant(1-w)C_k^S(\ket{\psi})   \right) \leqslant \exp \left( -\frac{w^2 N C_k^S(\ket{\psi})}{2}  \right)\leqslant \vartheta.
\end{split}
\end{equation}
Therefore, given confidence of $1-\vartheta$, one needs at least:
\begin{equation}
    N\geqslant \frac{2+w}{w^2 C_k^S(\ket{\psi})} \log\frac{1}{\vartheta},
\end{equation}
to estimate $C_k^S(\ket{\psi})$ with error $\pm wC_k^S(\ket{\psi})$.
The number of state copies is then:
\begin{equation}
    N_{\mathrm{tot}}=kN\geqslant \frac{(2+w)k}{w^2 C_k^S(\ket{\psi})} \log\frac{1}{\vartheta}.
\end{equation}
Then, consider the log error where:
\begin{equation}
    \left|\log\left( \widehat{C_k^S(\ket{\psi})} \right) - \log \left( C_k^S(\ket{\psi}) \right)\right| \leqslant \eta.
\end{equation}
Given confidence of $1-\vartheta$, we can set:
\begin{equation}
    \eta\leqslant\max \left\{ -\log(1-w), \log(1+w) \right\}.
\end{equation}
Therefore:
\begin{equation}
    w\leqslant 1 - e^{-\eta}.
\end{equation}
Therefore, the number of required state copies is then:
\begin{equation}
    N_{\mathrm{tot}}=kN \geqslant \frac{3-e^{-\eta}}{(1-e^{-\eta})^2C_k^S(\ket{\psi})}k\log\frac{1}{\vartheta}.
\end{equation}
For the multipartite cases, the conclusion remains the same as:
\begin{equation}
    N_{\mathrm{tot}}=kN_0\binom{n}{s} = kN \geqslant \frac{3-e^{-\eta}}{(1-e^{-\eta})^2C_k^s(\ket{\psi})}k\log\frac{1}{\vartheta}.
\end{equation}
As $(1-e^{-\eta})\sim\eta$ for small $\eta$, $\eta\sim O(N^{-1/2}_{\mathrm{tot}})$.

\subsection{Cyclic Permutation Test}

In this section we consider three scenarios for $\mathcal{G}=\mathcal{S},\mathcal{C},\mathcal{D}$, respectively. 
Firstly, we use the cyclic permutation test to estimate state moments
$\tr(\rho_S^{l})$ for multiple exponents $l$, thus estimating $C_k^S(\ket{\psi},\mathcal{S})$ from the gathered information of state moments.
Secondly, we use cyclic permutation test to directly estimate the acceptance probability for $\mathcal{G} = \mathcal{C}$.
Finally, by combining these two, namely, estimating both $\tr(\rho_S^2)$ and $C_k^S(\ket{\psi},\mathcal{C})$ respectively, we can then efficiently estimate $C_k^S(\ket{\psi},\mathcal{D})$.

\subsubsection{Estimating state moments}

Here we only consider the symmetric group, i.e., $\mathcal{G} = \mathcal{S}$, as this is a typical example to consider $\tr(\rho_S^l)$ for all $l$ from 2 to $k$ (and multiple $S$ in multipartite cases) as shown in Eq.~\eqref{eq:symmetric-part}.
From Lemma~\ref{lemma:mean-value}, we have shown that:
\begin{equation}
\left|C_k^S(\tau_2,\cdots,\tau_k) - C_k^S(\widehat{\tau_2},\cdots,\widehat{\tau_k})\right| \leqslant \sum_{l=2}^{k} \sup_{\xi\in[\mathbf{x},\mathbf{y}]}\left| \frac{\partial C_k^S}{\partial \tau_l}(\xi) \right| |\tau_l - \widehat{\tau_l}| \leqslant \sum_{l=2}^{k}\alpha_{l} \varepsilon_l.
\end{equation}
And for the symmetric group, we can set $\alpha_l = \frac{1}{l}$ as stated in Eq.~\eqref{eq:alpha-S}.
From~\cite{liu2025generalized} (cf. Supplemental Material I.B), we know that:
\begin{equation}
C_l^S(\ket{\psi}, \mathcal{C}) = \sum_{\sum_{x\in S}z_{x} \equiv 0 \bmod l}p(\ket{z_1\cdots z_n}) = \frac{1}{l}\sum_{q|l}\varphi(q)\tau_q^{\frac{l}{q}}.
\end{equation}
We denote:
\begin{equation}
J_0(l) = \sum_{\sum_{x\in S}z_{x} \equiv 0 \bmod l}p(\ket{z_1\cdots z_n}),
\end{equation}
therefore, while estimating $\tau_l$:
\begin{equation}
\widehat{\tau_l} = \frac{l\widehat{J_0(l)}-1 - \sum_{q|l,2\leqslant q \leqslant l-1}\varphi(q)\widehat{\tau_q}^{l/q}}{\varphi(l)}.
\end{equation}
We adopt Hoeffding's inequality again. 
For each time of the circuit execution, we acquire an outcome $\ket{z'_1\cdots z'_n}$. 
If $\sum_{x\in S}z'_{x} \equiv 0 \bmod K$, we set $M=1$. 
Otherwise, $M = 0$.
Therefore:
\begin{equation}
\widehat{J_0(l)} = \frac{1}{N_l}\sum_{i=1}^{N_l}M_i.
\end{equation}
Then:
\begin{equation}
p\left( |\widehat{J_0(l)} - J_0(l)| \geqslant \epsilon_l \right) \leqslant 2\exp(-2N_l\epsilon_l^2).
\end{equation}
Since:
\begin{equation}
\begin{split}
|\tau_l - \widehat{\tau_l}| = \varepsilon_l \leqslant& \frac{1}{\varphi(l)} \left( l|\widehat{J_0(l)} - J_0(l)| + \sum_{q|l,2\leqslant q \leqslant l-1}\varphi(q) \left|\widehat{\tau_q}^{l/q} - \tau_q^{l/q}\right|\right) \\
\leqslant & \frac{l}{\varphi(l)} \left(\epsilon_l + \sum_{q|l,2\leqslant q \leqslant l-1}\frac{\varphi(q)}{q} \varepsilon_q\right) \leqslant \cdots \\
\leqslant & \frac{l}{\varphi(l)} \left( \epsilon_l + \sum_{l'|l,2\leqslant l' \leqslant l-1}\epsilon_{l'}  + \sum_{l'|l,2\leqslant l' \leqslant l'-1}\sum_{l''|l',2\leqslant l'' \leqslant l'-1}\epsilon_{l''} + \cdots \right) \\
=&\frac{l}{\varphi(l)}\sum_{q|l, q\geqslant 2} c_{l,q} \epsilon_q,
\end{split}
\end{equation}
where $c_{l,q}$ denotes the number of divisor chains from $l$ down to $q$, i.e., the number of strictly decreasing sequences $\{l,l',\cdots,q\}$ with $l>l'>\cdots>q$ such that each smaller element divides the preceding larger one.
Then similarly, we set:
\begin{equation}
\left|C_k^S(\tau_2,\cdots,\tau_k,\mathcal{S}) - C_k^S(\widehat{\tau_2},\cdots,\widehat{\tau_k},\mathcal{S})\right| \leqslant \sum_{l=2}^{k}l^{-1} \varepsilon_l \leqslant \sum_{l=2}^{k}\frac{1}{\varphi(l)}\sum_{q|l, q \geqslant 2} c_{l,q} \epsilon_q = \varepsilon.
\end{equation}
Again, we split the failure rate $\delta$ to the estimation of each $J_0(l)$, then:
\begin{equation}
2\exp(-2N_l\epsilon_l^2) \leqslant \frac{\delta}{k-1} \ \Rightarrow \ N_l \geqslant \frac{1}{2\epsilon_l^2} \log\frac{2(k-1)}{\delta}.
\end{equation}
We optimize the number of copies:
\begin{equation}
N_{\mathrm{tot}}=\sum_{l=2}^{k} lN_l \geqslant \frac{1}{2}\log\frac{2(k-1)}{\delta}\sum_{l=2}^{k} \frac{l}{\epsilon_l^2},
\end{equation}
under the condition of:
\begin{equation}
\sum_{l=2}^{k}\frac{1}{\varphi(l)}\sum_{q|l, q \geqslant 2} c_{l,q} \epsilon_q = \sum_{q=2}^{k}\left(\sum_{\substack{l=2 \\ q| l,\, q\geqslant 2}}^{k}\frac{1}{\varphi(l)} c_{l,q}\right) \epsilon_q= \sum_{q=2}^{k}\beta_q\epsilon_q =\varepsilon.
\end{equation}
We again use Lagrange multiplier approach and we find the minimum of the lower bound is:
\begin{equation}
N_{\mathrm{tot}}=\sum_{l=2}^{k}lN_l \leqslant \frac{1}{2\varepsilon^2}\log\frac{2(k-1)}{\delta} \left(\sum_{l=2}^{k} l^{1/3} \beta_l^{2/3} \right)^3,
\end{equation}
with:
\begin{equation}
\epsilon_q = \left(\frac{2q}{\lambda \beta_q}\right)^{1/3}
\end{equation}
and:
\begin{equation}
\lambda = 2\left(\frac{\sum_{l=2}^{k}\beta_l^{2/3}l^{1/3}}{\varepsilon}\right)^3.
\end{equation}
Finally, note that $\beta_l \sim \tilde{\Theta}(1/l)$, for simplicity we adopt the scaling $1/l$, which makes the allocations $N_l$ coincide with those of the generalized SWAP test.
Again, for multipartite case, the scaling remains the same.
However, note that in this case, executing cyclic permutation circuits from $2$ to $k$ allows us to simultaneously acquire the information of both $C_k^S$ and $C_k^s$ for arbitrary $S$ or $s$.

For relative errors, we use the similar approach as shown before.
Due to the Chernoff bound:
\begin{equation}
\begin{split}
&p\left( \widehat{J_0(l)}\leqslant (1-w)J_0(l) \right) \leqslant \exp \left( -\frac{w^2 N_l J_0(l)}{2} \right) \\
&p\left( \widehat{J_0(l)}\geqslant (1+w)J_0(l) \right) \leqslant \exp \left( -\frac{w^2 N_l J_0(l)}{2+w} \right)
\end{split}
\end{equation}
We set $(1\pm w) J_0(l)  = J_0(l) \pm \epsilon_l$, then $w = \epsilon_l / J_0(l)$.
Therefore:
\begin{equation}
p\left( |\widehat{J_0(l)} - J_0(l)| \geqslant \epsilon_l \right) = p\left( \widehat{J_0(l)}\leqslant (1-w)J_0(l) \right) + p\left( \widehat{J_0(l)}\geqslant (1+w)J_0(l) \right) \leqslant 2\exp \left( -\frac{\epsilon_l^2 N_l}{2J_0(l)+1} \right) \leqslant \frac{\vartheta}{k-1}.
\end{equation}
Here we again split the failure probability equally. 
Therefore:
\begin{equation}
N_l \geqslant \frac{2J_0(l)+1}{\epsilon_l^2}\log\frac{2(k-1)}{\vartheta}.
\end{equation}
Similarly, once we set $\left|\log\left( \widehat{C_k^S(\ket{\psi}, \mathcal{S})} \right) - \log\left( C_k^S(\ket{\psi}, \mathcal{S}) \right)\right|\leqslant \eta$, then $\left|\widehat{C_k^S(\ket{\psi}, \mathcal{S})}  -  C_k^S(\ket{\psi}, \mathcal{S}) \right|\leqslant (1-e^{-\eta}) \max\{\widehat{C_k^S(\ket{\psi}, \mathcal{S})}, C_k^S(\ket{\psi}, \mathcal{S})\}$.
We set $\sum_{l=2}^{k}\frac{1}{\varphi(l)}\sum_{q|l,q\geqslant 2}c_{l,q}\epsilon_q = (1-e^{-\eta}) \max\{\widehat{C_k^S(\ket{\psi},\mathcal{S})}, C_k^S(\ket{\psi}, \mathcal{S})\}$ and we optimize:
\begin{equation}
N_{\mathrm{tot}}=\sum_{l=2}^{k}lN_l \geqslant \log\frac{2(k-1)}{\vartheta} \sum_{l=2}^{k} l\frac{2J_0(l)+1}{\epsilon_l^2}.
\end{equation}
Similarly, we use Lagrange multiplier approach and we finally have:
\begin{equation}
N_{\mathrm{tot}}=\sum_{l=2}^{k}lN_l \geqslant \log\frac{2(k-1)}{\vartheta}\frac{\left(\sum_{l=2}^{k}(2J_0(l)+1)^{1/3}l^{1/3}\beta_l^{2/3}\right)^3}{\left( (1-e^{-\eta}) \max\{\widehat{C_k^S(\ket{\psi}, \mathcal{S})}, C_k^S(\ket{\psi},\mathcal{S})\} \right)^2},
\end{equation}
with:
\begin{equation}
\epsilon_l = \left( \frac{(4J_0(l)+2)l}{\lambda \beta_l} \right)^{1/3},
\end{equation}
and:
\begin{equation}
\lambda = 2\left(\frac{\sum_{l=2}^{k}(2J_0(l)+1)^{1/3}l^{1/3}\beta_l^{2/3}}{(1-e^{-\eta}) \max\{\widehat{C_k^S(\ket{\psi},\mathcal{S})}, C_k^S(\ket{\psi},\mathcal{S})\}}\right)^3.
\end{equation}
Similarly, since we do not have prior knowledge of $J_0(l)$, we consider the worst scenario that $J_0(l)=1$, then:
\begin{equation}
N_{\mathrm{tot}}=\sum_{l=2}^{k}lN_l \geqslant \frac{3}{\left((1-e^{-\eta}) \max\{\widehat{C_k^S(\ket{\psi},\mathcal{S})}, C_k^S(\ket{\psi},\mathcal{S})\}\right)^2}\log\frac{2(k-1)}{\vartheta} \left( \sum_{l=2}^{k}l^{1/3}\beta_l^{2/3} \right)^3.
\end{equation}
As $(1-e^{-\eta})\sim\eta$ for small $\eta$, then again $\eta\sim O(N^{-1/2}_{\mathrm{tot}})$.

\subsubsection{Estimating $C_k^S(\ket{\psi}, \mathcal{C})$}

Again, for the parallelized cyclic permutation test circuit, it can be directly used to estimate $C_k^S(\ket{\psi},\mathcal{C})$ as:
\begin{equation}
    C_k^S(\ket{\psi}, \mathcal{C}) = \sum_{\sum_{x\in S}z_{x} \equiv 0 \bmod k}p(\ket{z_1\cdots z_n}),
\end{equation}
Similarly, we adopt Hoeffding's inequality. 
For each time of the circuit execution, we acquire an outcome $\ket{z'_1\cdots z'_n}$. 
If $\sum_{x\in S}z'_{x} \equiv 0 \bmod k$, we set $M=1$. 
Otherwise, $M = 0$.
Then for bipartite case, everything is the same as the proof in the previous secion.
Therefore, in order to get absolute error smaller than $\varepsilon$ with confidence $1-\delta$, one needs at least $\frac{k}{2\varepsilon^2}\log\frac{2}{\delta}$ copies of $\ket{\psi}$.
For the log error, given confidence of $1-\vartheta$, we can achieve $\left|\log\left( \widehat{C_k^S(\ket{\psi}, \mathcal{C})} \right) - \log \left( C_k^S(\ket{\psi}, \mathcal{C}) \right)\right| \leqslant \eta$ with at least $\frac{3-e^{-\eta}}{(1-e^{-\eta})^2\widehat{C_k^S(\ket{\psi}, \mathcal{C})}}k\log\frac{1}{\vartheta}$ copies of $\ket{\psi}$.

For the multipartite case, for each execution outcome $\ket{z_1\cdots z_n}$, we define:
\begin{equation}
    Y_i(\ket{z_1\cdots z_n}) = \frac{1}{\binom{n}{s}}\sum_{|S|=s}I^{(S)},
\end{equation}
where $I^{(S)}\in\{0,1\}$, denoting whether $\ket{z_1\cdots z_n}$ falls into the constraint $\sum_{x\in S}z_x\equiv0 \bmod k$.
For example, suppose $n,k=5$ and $s=2$ and for a specific execution label $i$, we have $\ket{z'_1\cdots z'_5}=\ket{32123}$, then $z'_1+z'_2=z'_1+z'_4=z'_2+z'_5=z'_4+z'_5\equiv0 \bmod 5$.
Therefore $Y_i(\ket{z'_1\cdots z'_5}) = \frac{1}{\binom{5}{2}}4=0.4$.
Then:
\begin{equation}
    \widehat{C_k^s(\ket{\psi},\mathcal{C})} = \frac{1}{N}\sum_{i=1}^{N}Y_i(\ket{z_1\cdots z_n}),
\end{equation}
and:
\begin{equation}
    \mathbb{E}\left[ Y_i \right] = C_k^s(\ket{\psi},\mathcal{C}).
\end{equation}
Then, again by using Hoeffding's inequality, we have:
\begin{equation}
    p\left( \left|\widehat{C_k^s(\ket{\psi},\mathcal{C})} - C_k^s(\ket{\psi},\mathcal{C})\right|\geqslant\varepsilon \right) = p\left( \left|\frac{1}{N}\sum_{i=1}^{N}Y_i(\ket{z_1\cdots z_n}) - \mathbb{E}\left[ Y_i \right]\right|\geqslant\varepsilon \right)\leqslant 2 \exp \left( -2N\varepsilon^2 \right) \leqslant \delta.
\end{equation}
Therefore, one also needs:
\begin{equation}
    N_{\mathrm{tot}}=kN \geqslant k\frac{\log \frac{2}{\delta}}{2\varepsilon^2}
\end{equation}
state copies to achieve absolute error smaller than $\varepsilon$ with confidence $1-\delta$.

For the relative error, since $e^{ax}$ is convex for any $a\in\mathbb{R}$, then for $x\in[0,1]$:
\begin{equation}
    e^{ax}\leqslant (1-x)e^0 + xe^{a}=1+x(e^a-1).
\end{equation}
Therefore:
\begin{equation}
    \mathbb{E}\left[e^{ax}\right]\leqslant1+(e^a-1)\mathbb{E}\left[ x \right].
\end{equation}
Since the random variable $0\leqslant Y_i \leqslant 1$, therefore:
\begin{equation}
    \mathbb{E}\left[e^{aY_i}\right]\leqslant1+(e^a-1)\mathbb{E}\left[ Y_i \right] = 1 + (e^a-1)C_k^s(\ket{\psi}, \mathcal{C}).
\end{equation}
Then, for $a>0$:
\begin{equation}
\begin{split}
    &p\left( \widehat{C_k^s(\ket{\psi}}, \mathcal{C})\geqslant (1+w)C_k^s(\ket{\psi}, \mathcal{C}) \right) \leqslant \inf_{a>0} \left(\exp\left( -a(1+w)NC_k^s(\ket{\psi},\mathcal{C}) \right) \prod_{i=1}^{N}\mathbb{E}\left[ \exp\left(aY_i(\ket{z_1\cdots z_n}\right) \right]\right) \\
    \leqslant& \inf_{a>0} \left(\exp\left( -a(1+w)NC_k^s(\ket{\psi},\mathcal{C}) \right) \left( 1+(e^a-1)C_k^s(\ket{\psi}, \mathcal{C}) \right)^N \right) \\
    \leqslant& \inf_{a>0} \left(\exp\left( -a(1+w)NC_k^s(\ket{\psi},\mathcal{C}) + N(e^a-1)C_k^s(\ket{\psi}, \mathcal{C}) \right) \right) \\
    =&\exp\left( -NC_k^s(\ket{\psi}, \mathcal{C}) \left( (1+w)\log(1+w)-w \right) \right),
\end{split}
\end{equation}
where the infimum is reached at $a = \log(1+w)$.
Similarly, we have:
\begin{equation}
    p\left( \widehat{C_k^s(\ket{\psi}}, \mathcal{C})\geqslant (1-w)C_k^s(\ket{\psi}, \mathcal{C}) \right) \leqslant \exp\left( -NC_k^s(\ket{\psi}, \mathcal{C}) \left( w+(1-w)\log(1-w) \right) \right).
\end{equation}
Since:
\begin{equation}
    \log(1+x) \geqslant \frac{2x}{2+x},
\end{equation}
for any $x>-1$. 
Then:
\begin{equation}
    p\left( \widehat{C_k^s(\ket{\psi}}, \mathcal{C})\geqslant (1+w)C_k^s(\ket{\psi}, \mathcal{C}) \right) \leqslant \exp \left( -NC_k^s(\ket{\psi}, \mathcal{C})\frac{w^2}{2+w} \right).
\end{equation}
Also since:
\begin{equation}
    -\log(1-x)\geqslant x+\frac{x^2}{2},
\end{equation}
for $0 \leqslant x < 1$.
Then:
\begin{equation}
    w+(1-w)\log(1-w)-\frac{w^2}{2}\geqslant 0.
\end{equation}
Then:
\begin{equation}
    p\left( \widehat{C_k^s(\ket{\psi}}, \mathcal{C})\geqslant (1+w)C_k^s(\ket{\psi}, \mathcal{C}) \right) \leqslant \exp \left( -NC_k^s(\ket{\psi}, \mathcal{C})\frac{w^2}{2} \right).
\end{equation}
Then the scaling for multipartite cases is again the same as the one in the previous secion. 
Therefore, given confidence of $1-\vartheta$, one needs at least:
\begin{equation}
    N_{\mathrm{tot}}=kN\geqslant \frac{3-e^{-\eta}}{(1-e^{-\eta})^2 C_k^s(\ket{\psi}, \mathcal{C})} k \log \frac{1}{\vartheta}
\end{equation}
state copies to reach the log error bound such that:
\begin{equation}
    \left|\log\left( \widehat{C_k^s(\ket{\psi}, \mathcal{C})} \right) - \log \left( C_k^s(\ket{\psi}, \mathcal{C}) \right)\right| \leqslant \eta.
\end{equation}
For small $\eta$, we have $(1-e^{-\eta})\sim\eta$.

\subsubsection{Estimating $C_k^S(\ket{\psi}, \mathcal{D})$}

From Supplemental Material~\ref{sm:exemplary-cases}, we know that for $\mathcal{G}=\mathcal{D}$:
\begin{equation}
\begin{split}
C_k^S(\ket{\psi}, \mathcal{D}) &= \frac{1}{2}C_k^S(\ket{\psi}, \mathcal{C}) + \frac{1}{4}\left( \left[\tr(\rho_S^2)\right]^{\frac{k-2+(k \bmod 2)}{2}} + \left[\tr(\rho_S^2)\right]^{\frac{k-(k \bmod 2)}{2}} \right)\\
&=\frac{1}{2}C_k^S(\ket{\psi}, \mathcal{C}) + \frac{1}{4}\left( (2C_2^S(\ket{\psi}, \mathcal{C})-1)^{\frac{k-2+(k \bmod 2)}{2}} + (2C_2^S(\ket{\psi}, \mathcal{C})-1)^{\frac{k-(k \bmod 2)}{2}} \right),
\end{split}
\end{equation}
and:
\begin{equation}
C_k^s(\ket{\psi}, \mathcal{D}) =\frac{1}{2}C_k^s(\ket{\psi}, \mathcal{C}) + \frac{1}{4\binom{n}{s}}\sum_{|S|=s}\left( (2C_2^S(\ket{\psi}, \mathcal{C})-1)^{\frac{k-2+(k \bmod 2)}{2}} + (2C_2^S(\ket{\psi}, \mathcal{C})-1)^{\frac{k-(k \bmod 2)}{2}} \right).
\end{equation}
Then we can estimate $C_k^S(\ket{\psi}, \mathcal{D})$ by executing two circuits. 
Firstly, we estimate $C_k^S(\ket{\psi},\mathcal{C})$ by performing $k$-copy parallelized cyclic permutation tests.
Secondly, we estimate $\tr(\rho^2_S)$ by performing $2$-copy parallelized swap tests.
Since:
\begin{equation}
\begin{split}
&\left| \widehat{C_k^S(\ket{\psi},\mathcal{D})} - C_k^S(\ket{\psi},\mathcal{D}) \right| \\
\leqslant& \frac{1}{2}\left|\widehat{C_k^S(\ket{\psi}, \mathcal{C})} - C_k^S(\ket{\psi}, \mathcal{C})\right| + \frac{1}{4}\left( 2\frac{k-2+(k \bmod 2) + k-(k \bmod 2)}{2} \left|\widehat{C_2^S(\ket{\psi}, \mathcal{C})} - C_2^S(\ket{\psi}, \mathcal{C})\right|\right)\\
=& \frac{1}{2}\varepsilon_k + \frac{1}{2}(k-1)\varepsilon_2,
\end{split}
\end{equation}
we set:
\begin{equation}
\frac{1}{2}\varepsilon_k + \frac{1}{2}(k-1)\varepsilon_2 = \varepsilon,
\end{equation}
Suppose we split the confidence interval equally for the errors $\varepsilon_k$ and $\varepsilon_2$, respectively.
Also, naively we can set $\varepsilon_k=\varepsilon$ and $\varepsilon_2 = \frac{\varepsilon}{k-1}$.
Then the total number of copies is:
\begin{equation}
N_{\mathrm{tot}}=kN_k + 2N_2 \geqslant \frac{k}{2\varepsilon_k^2}\log\frac{4}{\delta} + \frac{1}{\varepsilon_2^2} \log\frac{4}{\delta} = \frac{\log\frac{4}{\delta}}{\varepsilon^2} \left( \frac{k}{2} + (k-1)^2\right).
\end{equation}
One can also easily see that for multipartite case, the total required number of state copies remains similar.

For log errors, since we set:
\begin{equation}
\left| \log\widehat{C_k^S(\ket{\psi}, \mathcal{D})} - \log C_k^S(\ket{\psi}, \mathcal{D}) \right| \leqslant \eta,
\end{equation}
then:
\begin{equation}
\left| \widehat{C_k^S(\ket{\psi}, \mathcal{D})} - C_k^S(\ket{\psi}, \mathcal{D}) \right| \leqslant (1-e^{-\eta}) \max\{\widehat{C_k^S(\ket{\psi},\mathcal{D})}, C_k^S(\ket{\psi},\mathcal{D})\}.
\end{equation}
We now set:
\begin{equation}
\frac{1}{2}\varepsilon_k + \frac{1}{2}(k-1)\varepsilon_2 = (1-e^{-\eta}) \max\{\widehat{C_k^S(\ket{\psi},\mathcal{D})}, C_k^S(\ket{\psi}, \mathcal{D})\}.
\end{equation}
For simplicity, we choose $\varepsilon_k = (1-e^{-\eta}) \max\{\widehat{C_k^S(\ket{\psi},\mathcal{D})}, C_k^S(\ket{\psi},\mathcal{D})\}$ and $\varepsilon_2 = \frac{(1-e^{-\eta}) \max\{\widehat{C_k^S(\ket{\psi},\mathcal{D})}, C_k^S(\ket{\psi},\mathcal{D})\}}{k-1}$. 
Then we need:
\begin{equation}
N_{\mathrm{tot}}=kN_k + 2N_2 \geqslant \frac{\log\frac{4}{\vartheta}}{\left((1-e^{-\eta}) \max\{\widehat{C_k^S(\ket{\psi},\mathcal{D})}, C_k^S(\ket{\psi},\mathcal{D})\}\right)^2}\left(\frac{k}{2} + (k-1)^2\right),
\end{equation}
and we can easily acquire similar scaling for multipartite cases.
Note that optimizing the distribution of $\varepsilon_k$ and $\varepsilon_2$ can also be made by using Lagrange multiplier method as shown before, where we can distribute:
\begin{equation}
\varepsilon_k =  \frac{2k^{1/3}\varepsilon}{k^{1/3}+2^{1/3}(k-1)^{2/3}},
\end{equation}
and:
\begin{equation}
\varepsilon_2 = (k-1)^{-1/3}\frac{2^{4/3}\varepsilon}{k^{1/3}+2^{1/3}(k-1)^{2/3}}.
\end{equation}
This gives:
\begin{equation}
N_{\mathrm{tot}}=kN_k + 2N_2 \geqslant \log \frac{4}{\delta} \frac{(k^{1/3}+2^{1/3}(k-1)^{2/3})^2}{4\varepsilon^2} \left( \frac{k^{1/3}}{2} + \frac{(k-1)^{2/3}}{2^{2/3}} \right),
\end{equation}
thus yielding the allocations of $N_2$ and $N_k$ shown in Table~\ref{tab:alloc}.

\subsection{Remarks on State Moment Extrapolations}

The state moment $\tr(\rho^k)$ can be formulated in terms of the spectra of $\rho$, i.e., $\tr(\rho^k)=\sum_{i=1}^{r}\lambda_i^k$.
Once we acquire the information of each state moment from $\tr(\rho^2)$ to $\tr(\rho^r)$, any higher moment with $k>r$ can be extrapolated via the Newton–Girard method without requiring additional state copies.
This is discussed extensively in~\cite{shin2025resourceefficient}.
Accordingly, the numerics in Fig.~\ref{fig:fixed-Ntot} for $k>4$ are all extrapolated from the estimated results for $k=2,3,4$, since the acceptance probabilities for $k=2,3,4$ provide $\tr(\rho_S^2),\tr(\rho_S^3),\tr(\rho_S^4)$, and Haar random states are full rank, hence $r=2^{|S|}=4$ in our setting.

\begin{figure*}
    \centering
    \includegraphics[width=1.0\linewidth]{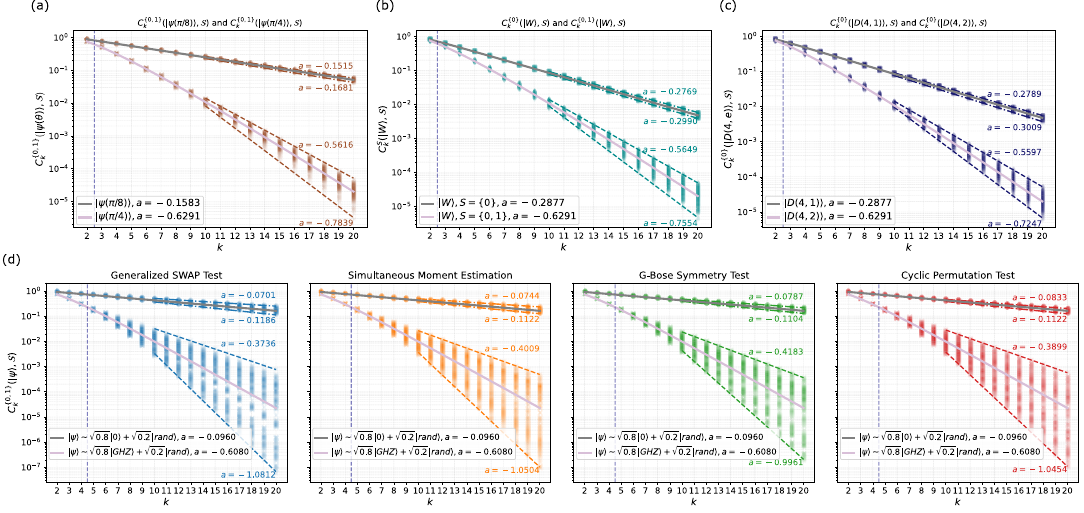}
    \caption{Numerical studies of acceptance probabilities under $\mathcal{S}$ and corresponding exponent fittings.
    For (a) $GHZ_{\theta}$ states, (b) $\ket{W}$ and (c) Dicke states, $10^5$ copies of the respective states are used to estimate $\tr(\rho_S^2)$, which is then extrapolated to acceptance probabilities for $k>2$ (right of the dashed line), as their nontrivial reduced states are always rank-2.
    Across 100 trials, the numerical estimates are plotted as scattered points, and the exponent $a$ is fitted in the form $e^{ak+b}$ for $k$ from 10 to 20.
    For (d), we compare the scenarios of $\ket{0}^{\otimes 4}$ and a 4-qubit $\ket{GHZ}$ perturbed by a Haar-random state $\ket{\mathrm{rand}}$.
    Since the reduced states of considered $\ket{\psi}$ in (d) are full rank (rank-4), their acceptance probabilities for $k=2,3,4$ are estimated by applying the four methods shown in Fig.~\ref{fig:circuits}.
    The extrapolations then start from $k=5$ (right of the vertical dashed line).}
    \label{fig:exponents}
\end{figure*}

\section{Examples\label{sm:example}}

In this section we illustrate several examples, including $GHZ_\theta$ states and Dicke states.
Since both classes are symmetric, we have $C_k^S(\ket{\psi})=C_k^s(\ket{\psi})$ whenever $|S|=s$.
In the derivations below we therefore present the multipartite case $C_k^s(\ket{\psi})$,
as the statements also apply accordingly to the bipartite case when $|S|=s$.
We also provide numerical studies in Fig.~\ref{fig:exponents}(a,b,c).
Since the cases in Fig.~\ref{fig:exponents}(a,b,c) are rank-2, estimating $\tr(\rho^2)$ via the SWAP test is sufficient.
For the full-rank setting, Fig.~\ref{fig:exponents}(d) compares the scenarios of the fully separable state $\ket{0}^{\otimes 4}$ and the 4-qubit $\ket{GHZ}$ state perturbed by a Haar-random state, where we extrapolate the acceptance probability from $k=5$.
This enables us to estimate the acceptance probability using the four different approaches listed in Fig.~\ref{fig:circuits}.
Finally, we restate the proof of the $k$-monotone decreasing behavior of $C_k^S(\ket{\psi}, \mathcal{S})$, i.e., $C_k^S(\ket{\psi}, \mathcal{S})\geqslant C_{k+1}^S(\ket{\psi}, \mathcal{S})$.

\subsection{$GHZ_{\theta}$ state}

We define $GHZ_{\theta}$ state as:
\begin{equation}
\ket{\psi(\theta)} = \sin \theta \ket{0}^{\otimes n} + \cos \theta \ket{1}^{\otimes n}
\end{equation}
For any of its reduced state $\rho_S$ with $s=|S|$:
\begin{equation}
    \rho_S = \sin^2 \theta \ket{0}^{\otimes s}\bra{0}^{\otimes s}+\cos^2 \theta \ket{1}^{\otimes s}\bra{1}^{\otimes s},
\end{equation}
which is rank-2 for $\sin\theta,\cos\theta\neq0$. 
From Eq.~\eqref{eq:partition}, we then have:
\begin{equation}
    C_k^s(\ket{\psi(\theta)}, \mathcal{S}) = \sum_{j=0}^{k} \binom{k}{j}\sin^{2j} \theta \cos^{2k-2j} \theta.
\end{equation}
If $\sin^2\theta \neq \cos^2\theta$:
\begin{equation}
    C_k^s(\ket{\psi(\theta)}, \mathcal{S}) = \frac{\sin^{2k+2}\theta - \cos^{2k+2}\theta}{\sin^2 \theta - \cos^2\theta}.
\end{equation}
If $\sin^2\theta = \cos^2\theta$:
\begin{equation}
    C_k^s(\ket{\psi(\theta)}, \mathcal{S}) = \frac{k+1}{2^k}.
\end{equation}
And note that:
\begin{equation}
    \lim_{\sin^2\theta-\cos^2\theta\rightarrow 0}\frac{\sin^{2k+2}\theta - \cos^{2k+2}\theta}{\sin^2 \theta - \cos^2\theta} =  \frac{k+1}{2^k}.
\end{equation}
It is also obvious that the minimum of $C_k^s(\ket{\psi(\theta)}, \mathcal{S})$ is $(k+1)2^{-k}$, achieved when $\sin^2\theta=\cos^2\theta=\frac{1}{2}$, i.e., when $\ket{\psi(\theta)}$ is $\ket{GHZ}$ up to local phases.
For the ratio between $C_{k+1}^s(\ket{\psi(\theta)}, \mathcal{S})$ and $C_{k}^s(\ket{\psi(\theta)}, \mathcal{S})$, we then have:
\begin{equation}
\lim_{k\rightarrow\infty}\frac{C_{k+1}^s(\ket{\psi(\theta)}, \mathcal{S})}{C_k^s(\ket{\psi(\theta)}, \mathcal{S})} = \max(\sin^2\theta,\cos^2\theta).
\label{eq:theta-lim}
\end{equation}
Also, for $\mathcal{G}=\mathcal{C}$:
\begin{equation}
C_k^s(\ket{\psi(\theta)}, \mathcal{C}) = \frac{1}{k}\sum_{q|k}\varphi(q)\left( \sin^{2q}\theta + \cos^{2q}\theta \right)^{k/q},
\end{equation}
and its minimum also reaches at $\sin^2\theta=\cos^2\theta=\frac{1}{2}$ and:
\begin{equation}
\min_{\theta}C_k^s(\ket{\psi(\theta)}, \mathcal{C}) = \frac{1}{k2^k}\sum_{q|k}\varphi(q) 2^{k/q}
\end{equation}
Finally, for $\mathcal{G}=\mathcal{D}$:
\begin{equation}
C_k^s(\ket{\psi(\theta)}, \mathcal{D}) = \frac{1}{2k}\sum_{q|k}\varphi(q)\left( \sin^{2q}\theta + \cos^{2q}\theta \right)^{k/q} + \frac{1}{4} \left( (\sin^4\theta+\cos^4\theta)^{\frac{k-2+(k \bmod 2)}{2}} + (\sin^4\theta+\cos^4\theta)^{\frac{k-(k \bmod 2)}{2}} \right),
\end{equation}
and its minimum again reaches at $\sin^2\theta=\cos^2\theta=\frac{1}{2}$ and:
\begin{equation}
\min_{\theta}C_k^s(\ket{\psi(\theta)}, \mathcal{D}) =\frac{1}{k2^{k+1}}\sum_{q|k}\varphi(q) 2^{k/q} + \frac{1}{4} \left(  2^{-\frac{k-2+(k \bmod 2)}{2}}+2^{-\frac{k-(k \bmod 2)}{2}} \right).
\end{equation}
We also list the true and estimated values of $C_k^{\{0,1\}}(\ket{\psi(\pi/8)}, \mathcal{S})$ and $C_k^{\{0,1\}}(\ket{\psi(\pi/4)}, \mathcal{S})$ for $k$ from 2 to 20, as shown in Fig.~\ref{fig:exponents}(a).
Using $10^5$ state copies, we report 100 trials for estimating $\tr\left(\rho_{\{0,1\}}^2\right)$ and then extrapolate.
We then fit the exponent $a$ in the form $e^{ak+b}$ for $C_k^{\{0,1\}}(\ket{\psi(\theta)}, \mathcal{S})$ over $k$ from 10 to 20.
Note that for $k\rightarrow\infty$, $a_{\mathrm{lim}}=\log(\max(\sin^2\theta,\cos^2\theta))$, according to Eq.~\eqref{eq:theta-lim}, which are approximately $-0.1583$ and $-0.6931$ for $\theta=\pi/8$ and $\pi/4$, respectively.

\subsection{W state}

For W state:
\begin{equation}
    \ket{W}=\frac{1}{\sqrt{n}}\left( \ket{10\cdots0}+\ket{01\cdots0}+\cdots+\ket{00\cdots1} \right),
\end{equation}
Then the reduced state $\rho_S$ has two eigenvalues $\frac{n-s}{n}$ and $\frac{s}{n}$, respectively, in the case of $s\neq 0,n$.
Then, if $s\neq \frac{n}{2}$:
\begin{equation}
C_k^s(\ket{W}, \mathcal{S}) = \frac{(n-s)^{k+1}-s^{k+1}}{n^k(n-2s)}.
\end{equation}
If $s=\frac{n}{2}$, then:
\begin{equation}
C_k^s(\ket{W}, \mathcal{S}) = \frac{k+1}{2^k} = C_k^s(\ket{GHZ}, \mathcal{S}),
\end{equation}
and similarly:
\begin{equation}
\lim_{s\rightarrow n/2}\frac{(n-s)^{k+1}-s^{k+1}}{n^k(n-2s)} = \frac{k+1}{2^k} = C_k^S(\ket{GHZ}, \mathcal{S}).
\end{equation}
As $C_k^s(\ket{W}, \mathcal{S})$ reaches minimum at exactly $s=\frac{n}{2}$ for fixed $k$ and $n$, therefore:
\begin{equation}
C_k^s(\ket{W}, \mathcal{S}) \geqslant C_k^s(\ket{GHZ}, \mathcal{S}) \ \Rightarrow \ \mathcal{E}_k^s(\ket{W}, \mathcal{S}) \leqslant \mathcal{E}_k^s(\ket{GHZ}, \mathcal{S}).
\end{equation}
Similarly:
\begin{equation}
\lim_{k\rightarrow \infty}\frac{C_{k+1}^s(\ket{W}, \mathcal{S})}{C_k^s(\ket{W}, \mathcal{S})} = \max\left(\frac{n-s}{n}, \frac{s}{n} \right).
\label{eq:lim-w}
\end{equation}

For $\mathcal{G}=\mathcal{C}$:
\begin{equation}
C_k^s(\ket{W}, \mathcal{C}) = \frac{1}{n^k k}\sum_{q|k}\varphi(q)\left( s^q + (n-s)^q \right)^{k/q},
\end{equation}
and its minimum reaches at $s=\frac{n}{2}$ and:
\begin{equation}
\min_{\theta}C_k^s(\ket{W}, \mathcal{C}) = \frac{1}{k2^k}\sum_{q|k}\varphi(q) 2^{k/q} = C_k^s(\ket{GHZ}, \mathcal{C}).
\end{equation}
Therefore:
\begin{equation}
C_k^s(\ket{W}, \mathcal{C}) \geqslant C_k^s(\ket{GHZ}, \mathcal{C}) \ \Rightarrow \ \mathcal{E}_k^s(\ket{W}, \mathcal{C}) \leqslant \mathcal{E}_k^s(\ket{GHZ}, \mathcal{C}).
\end{equation}
For $\mathcal{G}=\mathcal{D}$:
\begin{equation}
\begin{split}
C_k^s(\ket{W}, \mathcal{D}) =& \frac{1}{2n^k k}\sum_{q|k}\varphi(q)\left( s^q + (n-s)^q \right)^{k/q} + \\
& \frac{1}{4n^k}(s^2+(n-s)^2)^{\frac{k-2+(k \bmod 2)}{2}} + (s^2+(n-s)^2)^{\frac{k-(k \bmod 2)}{2}},    
\end{split}
\end{equation}
and its minimum reaches at $s=\frac{n}{2}$ and:
\begin{equation}
\min_{\theta}C_k^s(\ket{\psi(\theta)}, \mathcal{D}) =\frac{1}{k2^{k+1}}\sum_{q|k}\varphi(q) 2^{k/q} + \frac{1}{4} \left(  2^{-\frac{k-2+(k \bmod 2)}{2}}+2^{-\frac{k-(k \bmod 2)}{2}} \right) = C_k^s(\ket{GHZ}, \mathcal{D}).
\end{equation}
Therefore:
\begin{equation}
C_k^s(\ket{W}, \mathcal{D}) \geqslant C_k^s(\ket{GHZ}, \mathcal{D}) \ \Rightarrow \ \mathcal{E}_k^s(\ket{W}, \mathcal{D}) \leqslant \mathcal{E}_k^s(\ket{GHZ}, \mathcal{D}).
\end{equation}
Similarly, we list the true and estimated values of $C_k^{\{0\}}(\ket{W}, \mathcal{S})$ and $C_k^{\{0,1\}}(\ket{W}, \mathcal{S})$ (4-qubit $\ket{W}$) for $k$ from 2 to 20, as shown in Fig.~\ref{fig:exponents}(b).
We again use $10^5$ state copies for estimating $\tr\left(\rho_{\{0,1\}}^2\right)$, conduct 100 independent trials, and then fit the exponent $a$.
For $k\rightarrow\infty$, $a_{\mathrm{lim}}= \log\left(\max\left(\frac{n-s}{n}, \frac{s}{n} \right)\right)$, according to Eq.~\eqref{eq:lim-w}, which are approximately $-0.2877$ and $-0.6931$ for $S=\{0\}$ and $S=\{0,1\}$, respectively.

\subsection{Dicke state}

For Dicke state:
\begin{equation}
    \ket{D(n,e)}=\frac{1}{\binom{n}{e}}\sum_{ x\in\{0,1\}^{n}, \ w(x)=e } \ket{x}
\end{equation}
where $w(x)$ denotes the Hamming weight of string $x$.
Then the reduced state:
\begin{equation}
    \rho_S=\sum_{l= \max(0, e - n + s) }^{\min (s, e)} \frac{\binom{s}{l} \binom{n-s}{e-l}}{\binom{n}{e}} \ket{D(s,l)}\bra{D(s,l)},
\end{equation}
with the eigenvalues:
\begin{equation}
    \lambda_l =  \frac{\binom{s}{l} \binom{n-s}{e-l}}{\binom{n}{e}}.
\end{equation}
Suppose we fix $s=1$, then we have two eigenvalues (when $e\neq0,n$) $\frac{n-e}{n}$ and $\frac{e}{n}$, respectively.
Note that the expressions of the spectra are very similar to $\ket{W}$, by only replacing $s$ with $e$. 
Therefore, for $\mathcal{G}=\mathcal{S},\mathcal{C},\mathcal{D}$:
\begin{equation}
C_k^1(\ket{D(n,1)}) \geqslant C_k^1(\ket{D(n,2)}) \geqslant \cdots \geqslant C_k^1\left(\left|D\left(n,\left\lfloor\frac{n}{2}\right\rfloor\right)\right\rangle\right).
\end{equation}
Moreover, if $e=\frac{n}{2}$, for arbitrary non-trivial $s$ we have:
\begin{equation}
C_k^1\left(\left|D\left(n,e=n/2\right)\right\rangle\right) = C_k^{n/2}(\ket{W}) = C_k^{s}(\ket{GHZ}).
\end{equation}
The numerical exponent fittings for $\ket{D(4,1)}$ and $\ket{D(4,2)}$ are shown in Fig.~\ref{fig:exponents}(c).
The large-$k$ limit exponent is analogous to that of the $\ket{W}$ state, obtained by substituting $s$ with $e$ accordingly.

\subsection{Proof of $C_k^S(\ket{\psi}, \mathcal{S})\geqslant C_{k+1}^S(\ket{\psi}, \mathcal{S})$}

From Eq.~\eqref{eq:partition}, $C_k^S(\ket{\psi})$ can be expressed as a sum over partitioned spectrum powers.
Also, since $\sum_{i=1}^{r}\lambda_i=1$, we then have:
\begin{equation}
    C_k^S(\ket{\psi}, \mathcal{S}) = \sum_{\substack{g_1+\cdots+g_{r}=k \\ g_i\in\mathbb{Z} \ \mathrm{and} \ 0\leqslant g_i\leqslant k}} \lambda_1^{g_1} \cdots \lambda_r^{g_r} \cdot \sum_{j=1}^{r}\lambda_j = \sum_{j=1}^{r}\sum_{\substack{g_1+\cdots+g_{r}=k \\ g_i\in\mathbb{Z} \ \mathrm{and} \ 0\leqslant g_i\leqslant k}} \lambda_1^{g_1} \cdots \lambda_j^{g_j+1} \cdots \lambda_r^{g_r}.
\end{equation}
Then, we re-index each inner sum by defining $g_j'=g_j+1$ and $g_i'=g_i$ for $i \neq j$, then:
\begin{equation}
    C_k^S(\ket{\psi}, \mathcal{S}) = \sum_{j=1}^{r}\sum_{\substack{g'_1+\cdots+g'_j+\cdots+g'_{r}=k+1 \\ g'_i,g'_j\in\mathbb{Z}; \ 0\leqslant g'_i\leqslant k; \ 1\leqslant g'_j\leqslant k+1; \\ i\neq j}}\lambda_1^{g'_1} \cdots \lambda_j^{g'_j} \cdots \lambda_r^{g'_r}.
\end{equation}
Now we take any fixed $(g'_1,\cdots,g'_r)$ with sum $k+1$ and $0\leqslant\{g'_i\}_{i=1}^{r}\leqslant k +1$, and without constraints on $g'_{j}$ ($j\neq i$). 
Then in terms of the terms in $C_k^S(\ket{\psi}, \mathcal{S})$, this tuple appears at least once in the $j$-th outer sum if $g_j'\geqslant 1$, and it may also appear multiple times if several $g'_i\geqslant 1$.
Thus, every term in $C_{k+1}^S(\ket{\psi}, \mathcal{S})$ appears in $C_k^S(\ket{\psi}, \mathcal{S})$ at least once.
Therefore:
\begin{equation}
    C_k^S(\ket{\psi}, \mathcal{S}) \geqslant \sum_{\substack{g'_1+\cdots+g'_{r}=k+1 \\ g'_i\in\mathbb{Z} \ \mathrm{and} \ 0\leqslant g'_i\leqslant k+1}} \lambda_1^{g'_1} \cdots \lambda_r^{g'_r} = C_{k+1}^S(\ket{\psi}, \mathcal{S}).
\end{equation}
Based on this, we also have:
\begin{equation}
    \mathcal{E}_k^S(\ket{\psi}, \mathcal{S})\leqslant\mathcal{E}_{k+1}^S(\ket{\psi}, \mathcal{S}),
\end{equation}
and naturally:
\begin{equation}
    \mathcal{E}_k^s(\ket{\psi}, \mathcal{S})\leqslant\mathcal{E}_{k+1}^s(\ket{\psi}, \mathcal{S}).
\end{equation}

\end{document}